\documentclass[journal=jctcce,a4paper,manuscript=article]{achemso}

\usepackage{amsfonts}
\usepackage{amsmath}
\usepackage{amssymb}
\usepackage{amsthm}
\usepackage{braket}
\usepackage[american]{babel}
\usepackage{longtable}
\usepackage{float}
\usepackage{subfig}
\usepackage{graphicx} 
\usepackage{here}
\usepackage{multirow}
\usepackage{siunitx}
\usepackage{url}
\usepackage{xcolor}

\usepackage{hyperref}
\hypersetup{
    colorlinks = true,
    linkcolor = blue
    }
\usepackage[
        noabbrev,
        capitalise,
        nameinlink,
    ]{cleveref}

\newtheorem{theorem}{Theorem}
\newtheorem{lemma}{Lemma}

\setcitestyle{square}

\global\long\def\ket#1{\left|#1\right\rangle }
\global\long\def\bra#1{\left\langle #1\right|}

\global\long\def\ketbra#1#2{\ket{#1}\bra{#2}}

\title{Classification of electronic structures and state preparation for quantum computation of reaction chemistry}

\author{Maximilian M\"orchen}
\affiliation{ETH Zurich, Department of Chemistry and Applied Biosciences,\\ Vladimir-Prelog-Weg 2, 8093 Zurich, Switzerland}

\author{Guang Hao Low}
\affiliation{Microsoft Quantum, One Microsoft Way, Redmond, WA 98052, USA}

\author{Thomas Weymuth}
\affiliation{ETH Zurich, Department of Chemistry and Applied Biosciences,\\ Vladimir-Prelog-Weg 2, 8093 Zurich, Switzerland}

\author{Hongbin Liu}
\affiliation{Microsoft Quantum, One Microsoft Way, Redmond, WA 98052, USA}

\author{Matthias Troyer}
\affiliation{Microsoft Quantum, One Microsoft Way, Redmond, WA 98052, USA}
\email{matthias.troyer@microsoft.com}

\author{Markus Reiher}
\affiliation{ETH Zurich, Department of Chemistry and Applied Biosciences,\\ Vladimir-Prelog-Weg 2, 8093 Zurich, Switzerland}
\email{mreiher@ethz.ch}

\date{September 9, 2024}

\begin{document}
\maketitle

\begin{abstract}
Quantum computation for chemical problems will require the construction of guiding states with sufficient overlap with a target state. Since easily available and initializable mean-field states are characterized by an 
overlap that is reduced for multi-configurational electronic structures and even vanishes with
growing system size, we here investigate the severity of state preparation for reaction chemistry. We
emphasize weaknesses in current traditional approaches (even for weakly correlated molecules) and highlight the advantage of quantum phase estimation algorithms. 
An important result is the introduction of a new classification scheme for electronic structures based on orbital entanglement information. 
We identify two categories of multi-configurational molecules. 
Whereas class-1 molecules are dominated by very few determinants and often found in reaction chemistry, class-2 molecules do not allow one to single out a reasonably sized number of important determinants. The latter are particularly hard for traditional approaches and an ultimate target for quantum computation. Some open-shell iron-sulfur clusters belong to class 2. We discuss the role of the molecular orbital basis set and show that true class-2 molecules remain in this class independent of the choice of the orbital basis, with the iron-molybdenum cofactor of nitrogenase being a prototypical example. We stress that class-2 molecules can be build in a systematic fashion from open-shell centers or unsaturated carbon atoms.
Our key result is that it will always be possible to initialize a guiding state for chemical reaction chemistry in the ground state based on initial low-cost approximate electronic structure information, which is facilitated by the finite size of the atomistic structures to be considered. 

\end{abstract}

\section{Introduction}
The accurate simulation of quantum systems has been the original motivation for the development of the concept of a quantum computer \cite{Feynman1982}. Today we understand that simulating molecules and materials will not only be the first but the main  
application of quantum computers \cite{Hoefler2023}. 
Quantum computers are tremendously more complex than classical ones, and applications where they can outperform classical ones therefore need to focus on small-data problems with significant, and ideally exponential quantum speedup \cite{sho_polynomialtime_1997,2009AharonovJones,Abrams1997HubbardSimulation,Harrow2009}. 
The accurate simulation of molecules, which is of exponential classical complexity but does not require much data input or output is hence ideally suited.

However, more than just efficient simulation of a molecule is required to solve routine chemistry problems, of which catalytic processes have been highlighted as particularly interesting \cite{Reiher2016Reaction,vonBurg2020carbon,goings2022}. 
One of the most important tasks, the determination of ground state properties of a molecule by quantum phase estimation (QPE), requires in addition an approximate ground state $\ket{\psi_{\mathrm{approx}}}$ that has substantial overlap with the true ground state $\ket{\psi_{\mathrm{g}}}$~\cite{wg_efficient_2015}. The overlap determines the success probability 
\begin{equation}
\label{eq:p_success}
p_{\mathrm{success}} = 
\left| 
\langle \psi_{\mathrm{approx}} \psi_{\mathrm{g}} \rangle
\right|^2
\end{equation}
of the quantum algorithm reaching the ground state. 

However, Anderson~\cite{Anderson1967Catastrophe} famously noted an `orthogonality catastrophe' in many-body systems, wherein the overlap between the true ground state and an approximate product ansatz decreases exponentially with problem size. Although this is to be understood as being valid for extended systems in the asymptotic limit, not for molecules, it has been shown that vanishing overlaps can already be found for rather small finite-sized iron-sulfur clusters \cite{lee2023}, as they occur in electron transfer processes and reaction centers of metalloproteins.
The vanishing-overlap behavior is generic for growing system sizes, and then, the runtime grows exponentially also in the quantum case. Indeed, finding the ground states of generic quantum systems is QMA-complete~\cite{Aharonov2009QMALine,2002KitaevBook}; the quantum analog of NP-complete. 
Overall, this has suggested that there might not be an exponential quantum speedup for chemistry~\cite{lee2023}. However, despite the finite-cluster results that support this view, it strictly holds only in the thermodynamic limit. By contrast, we here consider problem sizes that relate to chemical reaction chemistry and are, due to the local nature of chemical reactions, limited to a comparatively small number of atoms.

In molecular strong-correlation
cases, the standard approach is to limit the wave function ansatz to an reduced-size $N$-(active)-orbital space, in which a subset $\eta$ of electrons can be correlated to yield a so-called full configuration interaction wave function in this limited space\cite{gao2024}.
The important question therefore is not the scaling of the overlap in the limit $\eta,N\to\infty$, 
but its value for typical molecules with moderate $N$. Wherever large overlap is found for a classically intractable instance, quantum computers retain the exponentially growing advantage for the value of $N$ in that instance. More important for quantum utility might be the fact that the rigorous error control of QPE algorithms already presents a key advantage over traditional approaches to enhance predictability of conclusions from simulations.\cite{Reiher2016Reaction,liu2022,vonBurg2020carbon}
However, this advantage only plays out if initial state preparation does not compromise the quantum simulation by significantly enlarging the runtime due to low overlap.

In this paper, we first clarify the situation for routine reaction chemistry applications and emphasize that elementary reaction steps involve rather few atoms so that a reliable modelling of them requires on the order of about 100 atoms. This naturally limits the effects of the orthogonality catastrophe and allows us to then 
analyze approximate ground-state wave functions of electronic structures that feature typical multi-configurational problems. As a result, we introduce two classes of multi-configurational molecules based on single-orbital entropy analyses --- alongside a third class for weakly correlated electronic structures --- and assess the overlap of easily constructible trial states for them. Finally, we comment on state preparation that can cope with the two multi-configurational classes
(and discuss procedures for this in the appendix).

\section{Quantum Computing for Reaction Chemistry}

Understanding chemical reaction chemistry requires the elucidation of reaction mechanisms and the microkinetic modelling of concentration fluxes through these chemical reaction networks of stable intermediates and transition states (see Refs.\ \citenum{Bensberg2023a, Bensberg2024} for an example). 
This can be accomplished by quantum mechanical optimizations of the underlying molecular structures and the assignment of an energy to these species, possibly in a largely automated fashion in order to cover as much of the relevant reaction space as possible (see Refs.\ \citenum{Sameera2016, Dewyer2017, simm2019, steiner2022, Ismail2022, Wen2023, Margraf2023} for reviews).
Reliability and predictability --- which are key for
chemical optimization and design attempts --- depend then inevitably on the 
accuracy of the energies assigned to the species in the reaction network. However, the precise accuracy of an individual energy for a specific molecular structure will, in general, not be known from any traditional quantum chemistry approach, only statistical knowledge may guide the reliability assessment\cite{reiher2022}.
We have therefore advocated \cite{Reiher2016Reaction,liu2022,vonBurg2020carbon}
the well-defined error estimation of QPE algorithms in fault-tolerant quantum computation to be
a key advantage over traditional quantum chemical calculations. 

It is important to note that traditional approaches in quantum chemistry do not provide information about the accuracy of a specific result obtained for a specific molecular structure.\cite{reiher2022} While some parameters can actually be well controlled in such calculations (such as the size of the one-electron (orbital) basis set through extrapolation to the orbital basis limit), it is the ansatz of the approximate many-electron wave function that introduces an error that cannot be assessed in a system-focused fashion by a method-inherent approach (and a comparison to a more accurate approximation is hardly available; if it were, one would resort to these more accurate data, but again not knowing the precise accuracy obtained). By contrast, QPE provides a controllable error for a specific calculation. That is, it provides an error estimate for the specific molecular structure under consideration and this error can be systematically made smaller as the runtime is increased.

Although the focus of quantum computing for chemical problems has often (and for good reason) been skewed to truly challenging electronic structure problems (such as those in Refs.\ \citenum{Reiher2016Reaction,goings2022}), 
a true chemical breakthrough will only be accomplished if quantum computation becomes routinely available for chemical problems. 

The error estimation approach of QPE implies that we can obtain accurate relative energies for chemical reactions and activation barriers from total electronic energies with controllable error (surely, vibrational and other thermal contributions also matter, but it is foremost the potentially inaccurate electronic structure model that must be fixed before further corrections are worth being considered). 

However, since relative electronic energies are not directly accessed with QPE, we cannot benefit from systematic error compensation, which is key to traditional approaches. A severe consequence of this observation is that the QPE advantage can be compromised by the orthogonality catastrophe because the overlap of a target state with a low-cost wave function such as the mean-field  (Hartree--Fock) determinant
will decay for increasing molecule sizes. 

While this only holds true for growing system sizes, it must be understood that reaction chemistry is a local phenomenon, usually involving the breaking of one or two bonds at a time. Hence, atomistic models describing bond dissociation and formation can be confined to a limited, comparatively low number of atoms. 
Large molecular structures exhibit a huge fraction of atoms that are spectators and can be efficiently described by embedding methods (see, for instance, Ref.\ \citenum{Erakovic2024}). 
As a result, a sufficiently large structural model of a reactive event only needs to comprise about 100 to 200 atoms. Many examples exist in the computational chemistry literature that demonstrate this point and we may simply refer to the modeling approach developed by Siegbahn\cite{siegbahn2011}
and to the idea of quantum-classical hybrid modeling in general
\cite{senn2007}.
Therefore, the reaction chemistry problem does not suffer from the problematic scaling of a vanishing overlap with increasing system size.

\section{A New Classification Scheme for Electronic Structures}

Even if we can assume a fixed molecule size, the overlap of the Hartree--Fock (HF) wave function with the target state might still be small (say, 0.1), but this would still be sufficient to run QPE algorithms efficiently.
However, one may encounter molecular problems with less than 200 atoms that exhibit strongly correlated ground states. Whereas such systems are truly scarce in thermal equilibrium chemistry, excited states, relevant to photophysics and photochemistry, and open-shell transition metal clusters often feature strong electron correlation. Hence, such situations require a rigorous analysis of the guiding state preparation problem. For this reason, we begin with the introduction of three categories of electronic structures, including two different classes of multi-configurational molecules (see Table \ref{classes}).

\begin{table}
    \centering
    \caption{\label{classes}Definition of three classes of electron-correlation problems according to the contribution of determinants to the FCI wave function expansion (middle column). The last column provides a description in terms of standard quantum chemical terminology for electron correlation.}
    \begin{tabular}{lll}
        \hline\hline
        class & characterization & correlation types\\
        \hline
        0 & HF determinant features high weight & exclusively dominated by \\
          &                                     & dynamic correlation \\[2ex]
        1 & few (far less than 100) determinants & static correlation; few \\
          & have similar non-negligible weight        & orbitals are important \\[2ex]
        2 & no single or small set of determinants & strong static correlation; \\
          & have a weight significantly            & all valence orbitals are \\
          & larger than that of other determinants & strongly entangled \\
        \hline\hline
    \end{tabular}
\end{table}

Class 0 collects all cases, in which the HF determinant is an excellent starting point for constructing a reliable approximation to the electronic wave function, as it is accomplished by traditional single-reference coupled cluster theory.
In class 1, we find those cases that are dominated by a few Slater determinants with significant weight.
Such situations may be found in stable intermediates, but they more often occur during a reaction because of the bond breaking and bond forming processes in an elementary reaction step and can therefore be found in transition state structures.
Hence, these are cases routinely encountered in reaction mechanisms. 
In class 2, we find extreme cases of strong static electron correlation, where no determinant features a weight in the determinantal expansion of the ground state that is significantly larger than that of any other determinant. 

While this classification has so far been based on somewhat hand-waving intuition, we wish to point out that
this classification scheme is mirrored in entropy measures for orbital entanglement that we depict in so-called threshold diagrams \cite{Stein2016}, as we shall see in the following.
In these threshold diagrams, the $x$-axis lists single-orbital entropies taken as thresholds to count how many orbitals are assigned this entropy or a larger value of it (the number given on the $y$-axis; see the appendix for the definition of these entropy measures and further details on threshold diagrams). 

Fig.\ \ref{fig:ideal_thresholds-results} shows, in the top panel, idealized orbital-entropy threshold diagrams for molecules of class 0, 1, and 2.
For class-0 molecules, no entropy is larger than a given threshold, and hence, they are interpreted as single-configurational. This is expressed by the initial steep drop in orbital number in these diagrams.
For class-1 electronic structures, 
plateaus emerge, which indicate orbitals with a similar degree of orbital entanglement that is exploited to choose active spaces in the AutoCAS algorithm \cite{Stein2016}. For class-2 molecules, no plateaus can be identified as almost all valence orbitals are strongly entangled.
The lower panel in Fig.\ \ref{fig:ideal_thresholds-results} presents explicit threshold diagrams for water, the transition state structure of cyclobutadiene isomerization, and the active site of iron-molybdenum nitrogenase, i.e., the FeMo-cofactor (FeMoco)
as examples for the three classes.

\begin{figure}
    \centering
    \includegraphics[width=\textwidth]{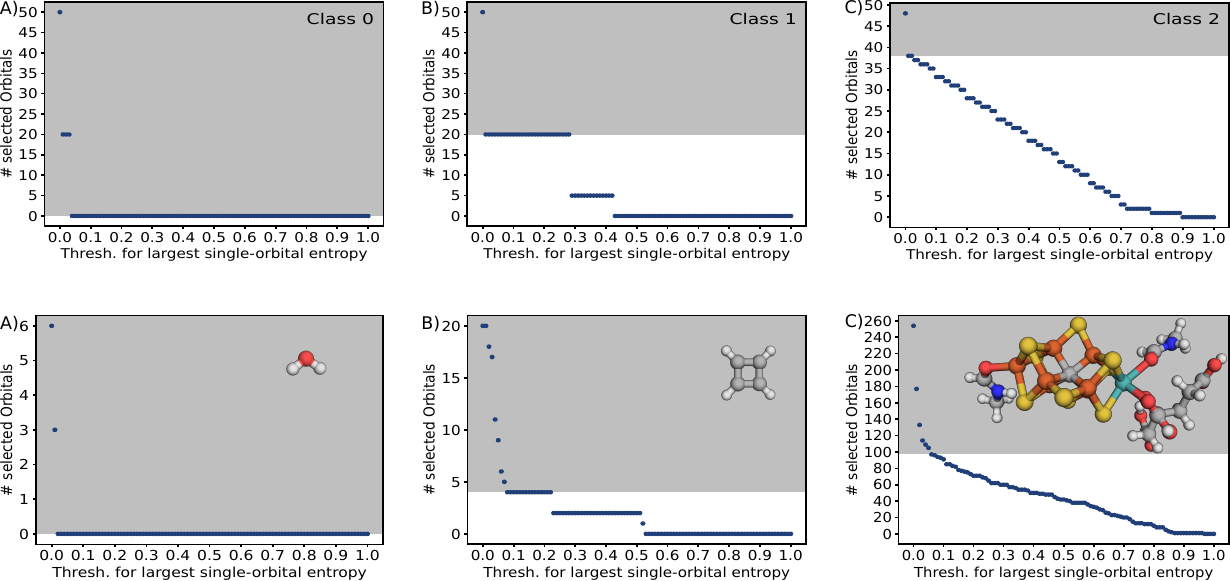}
    \caption{Orbital-entropy threshold diagrams for the different correlation classes of electronic structures obtained in all-valence complete active space calculations. Top panel: Idealized diagrams that characterize class-0, class-1, and class-2 molecules.
The gray rectangle denotes the number of orbitals that are assigned to be dynamically correlated (i.e., weakly entangled) and that would therefore be excluded from an active orbital space of strongly correlated orbitals.
Lower panel: Calculated diagrams obtained for water as a class-0 case, the transition state structure of cyclobutadiene isomerization as a class-1 case, and FeMoco (in a split-localized basis) as a class-2 case.
Color code: yellow denotes sulfur atoms, orange iron atoms , red oxygen atoms, gray carbon atoms, dark blue nitrogen atoms, white hydrogen atoms, and turquoise molybdenum.
    See the appendix for details on the computational methodology.}
    \label{fig:ideal_thresholds-results}
\end{figure}

We note that the assignment of the electronic structure of a molecule to one of these three classes cannot be stringent, and borderline cases will exist.
Still, the introduction of these three classes serves the purpose to address different challenges for the potential benefit of applying quantum computation to obtain electronic energies and for the state preparation procedure.
Whereas molecules in class 0 can always be initialized with the mean-field Hartree--Fock state, molecules in class 1 will benefit from the initialization of a few, rather than one determinant in order to reduce the effort for the QPE step.
Molecules in class 2 will require significantly more effort for state preparation, which could require initialization of millions of determinants or of some other format such as matrix product states (MPSs).

In general, one will not know to what class a molecular structure will be assigned to (except for prototypical, well-understood systems; but consider those in Ref.\ \citenum{Stein2017b} as counter example) since the composition of the target state is not known a priori. 
However, one may carry out exploratory configuration interaction (CI) singles-doubles or small-bond-dimension density matrix renormalization group calculations to gain some initial information with limited effort. In fact, the latter choice has been demonstrated to work remarkably well for the selection of active orbitals spaces \cite{Stein2016,Stein2016b,Stein2017,Stein2017b,Stein2019}.

\section{State Preparation Complexity}
From the point of view of initial state preparation for QPE, electronic structures belonging to class 0 are trivial to prepare since the mean-field Hartree--Fock state will solely dominate the determinant expansion of the target state. Although multi-configurational molecules are more challenging in this respect, we note that all class-1 molecules are still straightforward to intialize because their electronic structures are dominated by a rather small number of determinants. Only class-2 molecules are challenging and hence interesting to study from this point of view.

Recently it has been argued\cite{lee2023} that quantum computation for chemical problems can be severely hampered by the construction of guiding states, even for rather small molecules that show strong electronic correlation.
These small molecules, however, namely iron-sulfur clusters, are likely to belong to class 2.
We will therefore study those systems in this work, but note that they do not play a key role in general reaction chemistry as reaction centers (for instance, they are sensitive against reactive oxygen species \citenum{Stiebritz2012})
and rather operate as biomolecular electron-transfer relays (i.e., as 'sponges' that take up or release electrons in enzymes).
We also note that the findings in Ref.\ \citenum{lee2023} were in contrast to the earlier claims by Tubman et al.\cite{tubman2018postponing} that the Hartree--Fock determinant or a multi-configurational guiding state will be efficient for QPE algorithms.

If the weight of the Hartree--Fock determinant in the determinantal expansion of the total state is larger than about 0.9, the wave function is considered to be single-configurational\cite{jiang2012, langhoff1974, lee1989, sears2008a, sears2008b} (and all remaining correlation is called dynamic correlation in quantum chemistry).
For molecular systems with this property (those that do neither belong to class 1, nor to class 2 and may therefore be collected in class 0), it was emphasized\cite{lee2023} that standard coupled cluster theory on classical computers already provides sufficiently accurate results. 
While coupled cluster theory can provide results close to FCI for small molecules such as dinitrogen considered in Ref.\ \citenum{lee2023} where one can push the excitation rank to five- and even six-fold excitations (albeit in a double-zeta basis and for a frozen core\cite{Chan2004}), this is not at all possible for large molecules.
For large molecules of class 0, Ref.\ \citenum{lee2023} claims the standard CCSD(T) model to be a  sufficiently accurate model, although explicit knowledge about its accuracy in a specific case will not be known.
On the contrary, even for closed-shell systems the CCSD(T) error may be insufficient for reliable kinetic modeling \cite{Karton2024} and for transition-metal complexes one may expect errors that are on average 3 kcal/mol \cite{DeYonker2007}. 
Moreover, large molecules can only be treated by local coupled cluster variants which come with new approximations and thresholds\cite{riplinger2016,Qianli2018} whose effect on growing system sizes remains unclear (see, for example, deficiencies in the DLPNO-CCSD(T) model \cite{altun2021} to which Ref.\ \citenum{lee2023} refers to for large molecules such as small proteins).
Again, QPE with rigorous error estimation for fixed-size reactive molecular (sub)systems promises a way out of this situation.

In the following, we investigate the size of active orbital spaces for molecules belonging to the three classes of orbital entanglement, which then allows us to discuss the severeness of the guiding state preparation problem for (reaction) chemistry.
Determinant weights (i.e., the squares of CI expansion coefficients $c_i$) and the single-orbital entropies are taken as descriptors for the analysis of the electronic wave function.
We investigate how much information on the target state can already be extracted from the approximate MPS generated for the selection of the active orbital space with the AutoCAS algorithm\cite{Stein2016}. 

\subsection{Class-0 Cases}

For comparison, we show results for the class-0 case of water\cite{Hoy1979} in Table \ref{tab:class1}.
Note that the wave function expansion of the class-0 water molecule is dominated solely by the HF determinant. 
In this case, it is neither possible nor sensible to define an active space.
Although the corresponding HF energy is not of sufficient accuracy for chemical applications, 
the HF determinant is clearly very well suited for guiding-state preparation (first row in Table \ref{tab:class1}).
Class-0 cases are therefore trivial in this respect.

\subsection{Class-1 Cases}

We evaluated the electronic structure for molecules that will show multi-configurational character and that we may expect to find in class 1 according to previous results: N$_2$O$_4$\cite{pulay1988}, C$_4$H$_4$\cite{lyakh2011}, O$_3$\cite{pulay1988}, $^1$C$_6$H$_4$\cite{evangelista2012}, $^3$C$_6$H$_4$\cite{evangelista2012}, $^1[$Fe(NO)(CO)$_3]^-$\cite{sayfutyarova2017} and $^3[$Fe(NO)(CO)$_3]^-$\cite{sayfutyarova2017}. 
Table \ref{tab:class1} shows the results of the initial orbital selection based on approximate density matrix renormalization group (DMRG) data for the full valence space (left hand side of the arrow in Table \ref{tab:class1}).

\begin{table}[htp!]
    \centering
    \begin{tabular}{c|c c c c  c c c c}\hline\hline
System & CAS(e, o)$^\text{ini}$ & $N_d^\text{ini}$ & $A_w^\text{ini}$ & $\rightarrow$ & CAS(e, o)$^\text{fin}$ & MPS overlap & $N_d^\text{fin}$ & $A_w^\text{fin}$ \\
\hline                                                         
H$_2$O & (8, 6)       & 1 & 0.996 & $\rightarrow$ & (0, 0) & --- & 1 & --- \\
\hline                                      
\hline                                      
                     && 1 & 0.811 &&&                                 & 1 & 0.905 \\
N$_2$O$_4$ & (34, 24) & 2 & 0.837 & $\rightarrow$ & (6, 6) & 1.000000 & 2 & 0.928 \\
                     && 4 & 0.843 &&&                                 & 3 & 0.940 \\
\hline                                      
                     && 1 & 0.443 &&&                                 & 1 & 0.465 \\
C$_4$H$_4$ & (20, 20) & 2 & 0.886 & $\rightarrow$ & (4, 4) & 1.000000 & 2 & 0.929 \\
                     && 4 & 0.919 &&&                                 & 4 & 0.971 \\
\hline                                      
                    && 1 & 0.818 &&&                                  & 1 & 0.868 \\
$^1$O$_3$ & (18, 12) & 2 & 0.903 & $\rightarrow$ & (4, 5)  & 1.000000 & 2 & 0.981 \\
                    && 4 & 0.913 &&&                                  & 3 & 0.991 \\
\hline                                      
                         &&  1 & 0.611 &&&                                 & 1 & 0.682 \\
$^1$C$_6$H$_4$ & (28, 28) &  2 & 0.844 & $\rightarrow$ & (6, 6) & 1.000000 & 2 & 0.934 \\
                         && 13 & 0.871 &&&                                 & 3 & 0.939 \\
\hline                                      
                         && 1 & 0.872 &&&                                 & 1 & 0.899 \\
$^3$C$_6$H$_4$ & (28, 28) & 2 & 0.883 & $\rightarrow$ & (8, 8) & 1.000000 & 2 & 0.916 \\
                         && 7 & 0.890 &&&                                 & 6 & 0.950 \\
\hline                                      
                                 &&   2 & 0.272 &&&                                   &    1 & 0.769 \\
$^1[$Fe(NO)(CO)$_3]^-$ & (50, 46) &   6 & 0.511 & $\rightarrow$ & (10, 12) & 0.999936 &    7 & 0.805 \\
                                 && 104 & 0.700 &&&                                   &  325 & 0.900 \\
\hline                                      
                                 && 1 & 0.585 &&&                                    &   1 & 0.833 \\
$^3[$Fe(NO)(CO)$_3]^-$ & (50, 46) & 2 & 0.618 & $\rightarrow$ & (10, 12)  & 0.999999 &   5 & 0.860 \\
                                 && 9 & 0.700 &&&                                    &  45 & 0.900 \\
\hline\hline
    \end{tabular}
    \caption{The \textit{initial} ('ini') and \textit{final} ('fin') active space information for class-1 molecules and for water (class 0, first row). 
    $N_d^\text{ini}$ and $A_w^\text{ini}$ measure the number of significantly contributing Slater determinants and their corresponding accumulated weight in the approximate MPS obtained for the \textit{initial} (full-valence) complete active space (CAS), respectively.
    'MPS overlap' is evaluated for the approximate and converged MPSs for the \textit{final} active space.
    $N_d^\text{fin}$ corresponds to the number of determinants to evaluate the accumulated weight $A_w^\text{fin}$, also for the \textit{final} active space.
    Note that the determinants counted in $N_d^\text{ini}$ and $N_d^\text{fin}$ are not necessarily the same. 
    }
    \label{tab:class1}
\end{table}

Apart from the water example, all data in Table \ref{tab:class1} classify almost the other molecules as multi-configurational (more than one determinant contributes to the wave function, and often far less than 10 determinants are required to reach an accumulated weight of above 0.92). Some molecules have a weight of the HF determinant close to 0.9 ($^1$O$_3$, $^3$C$_6$H$_4$, and both Fe nitrosyl complexes) or even slightly larger (N$_2$O$_4$).
Clearly, all molecules in Table \ref{tab:class1} are of class 1 (except water, of course) and can be considered to represent typical multi-configurational cases of a type also encountered in reaction chemistry. 
Still, it is obvious that all equilibrium structures can be initialized by the HF determinant alone. 
This is in line with findings by Tubman et al.\cite{tubman2018postponing}
For the C$_4$H$_4$ transition state, the HF determinant does not have a significant weight, but instead two determinants arising by single excitations from the HF determinant 
yield an accumulated weight that is larger than 0.9 (this weight can also be accomplished by combining both determinants into a single configuration state function). 
For state preparation, one may simply initialize the few determinants that enhance the accumulated weight to a value of 0.92 and higher.
We note that all of these class-1 molecules show a typical plateau structure in the threshold diagrams defined in Ref.\ \citenum{Stein2016}. 

As a large iron-centered porphyrine complex, one would expect the complex abbreviated as Fe-Porph (see the inlay in Fig.\ \ref{fig:fe_porph} for its molecular structure) to be a class-2 system.  
But even such potential class-2 systems can actually be class-1 molecules,
for which Fe-Porph turned out to be an example. 
Due to the large number of valence orbitals, the active space can only be selected through the large active space protocol introduced in Ref.\ \citenum{Stein2019}. 
The threshold diagram in Fig.\ \ref{fig:fe_porph} (left) suggests a relatively small active space of only 10 orbitals, which is of the same size as that of other class-1 single-center metal complexes. 
We validated the active space choice by optimizing a smaller active space of 30 orbitals for Fe-Proph with DMRG. 
The resulting threshold and orbital entanglement diagrams are shown in Fig.\ \ref{fig:fe_porph} (right).
Orbital entanglement diagrams deliver a graphical representation of the orbital entropies for each orbital arranged as a circle. All orbitals are pairwise connected by a line whose thickness encodes the value for the corresponding  mutual information.

\begin{figure}
    \centering
    \includegraphics[width=0.45\linewidth]{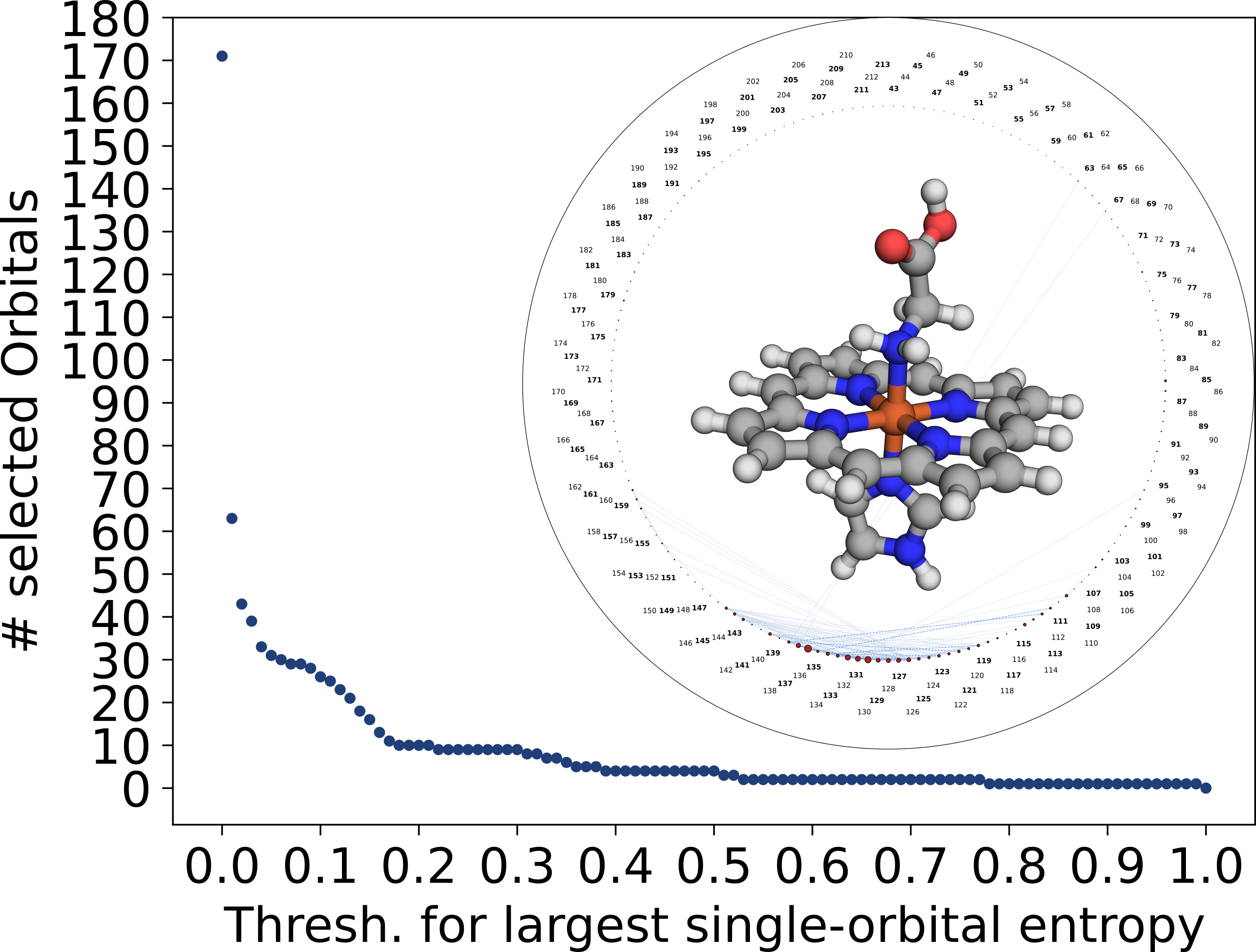}
    \includegraphics[width=0.45\linewidth]{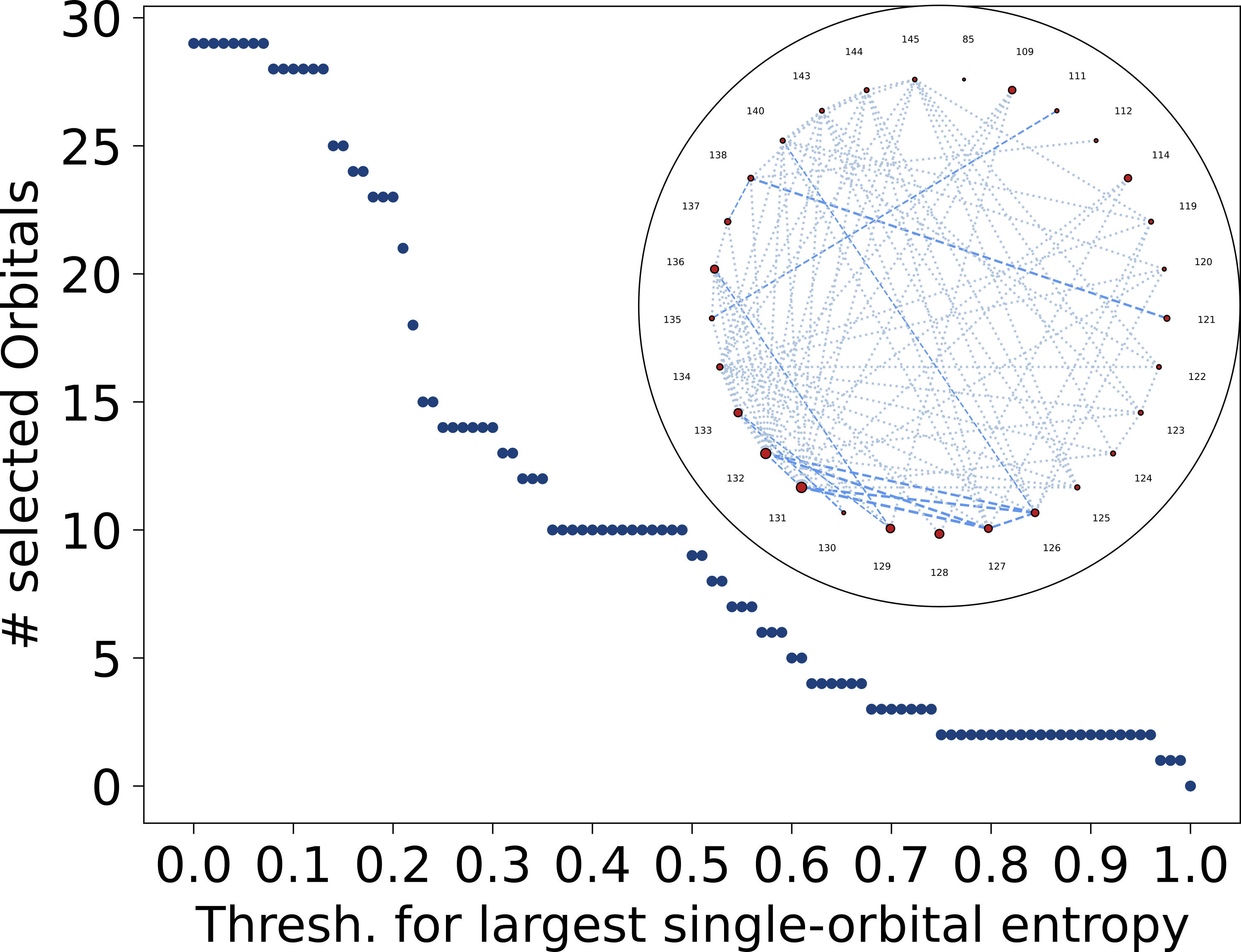}
    \caption{Left: Full-valence threshold diagram of the Fe-Porph complex based on large CAS (30 orbitals per sub-CAS; the united orbital entanglement diagram and the optimized structure of the complex are shown as an inlay).
    The first (small) plateau is found at 10 orbitals, followed by a larger plateau at 9 orbitals. 
    The corresponding active spaces both include 10 electrons.
    Right: 
Threshold and orbital entanglement (inlay) diagrams of a small active space of 30 orbitals and 34 electrons for the Fe-Porph complex.
    Now, a large plateau is clearly visible at 10 orbitals. 
    }
    \label{fig:fe_porph}
\end{figure}

We now turn to other potential class-2 molecules, where we demonstrate that this drastic reduction in active space is not always found.

\subsection{Orbital Effects and Class-2 Cases\label{sec:orbitals}}

In contrast to class-1 molecules, the valence space of most class-2 molecules is larger and more orbitals are stronger correlated. However, the
orbital choice can have a significant effect on orbital entanglement, and hence, on the selection of active orbitals.
We therefore investigated the effects of different orbital choices (see appendix for details). 
For [2Fe-2S] with a valence space of about 80 orbitals, we found plateaus in the threshold diagrams that would lead to comparatively small active spaces of
less than 20 orbitals for all orbital sets that we considered, hence indicating that this small iron-sulfur cluster
can be considered a class-1 case.
However, this picture will change as the FeS clusters become bigger owing to the increasing number of open-shell metal centers.
The larger iron-sulfur cluster [4Fe-4S] features an
initial active space of 110 electrons and 116 orbitals. 
In view of the complete active space (CAS) selections found for the smaller iron-sulfur cluster [2Fe-2S], orbital localization is key to reduce orbital entanglement and eventually allows one to reduce the active space size. This was also found for the [4Fe-4S] cluster, where the final active space in a split-localized basis comprised 43 orbitals, making it a boundary case between classes 1 and 2.

For our FeMoco structural model, the initial valence space contains 249 electrons in 254 orbitals (CAS(249, 254)).
It was already recognized by Chan and co-workers \cite{Li2019} that a full valence space for FeMoco is likely to be necessary to represent the S=3/2 resting state of FeMoco.
However, these authors proposed a significantly smaller active space of 113 electrons in 76 spatial orbitals (including Fe $3d$, S $3p$, Mo $4d$, C $2s2p$, and ligand orbitals) than what we consider in this paper.
The reason for this large discrepancy is likely the fact that we include the double-$d$ shell for all metal centers and also the valence-shell contributions of all atoms in our FeMoco model. Also note that we have different ligands coordinating, which is, however, not decisive for the fact that the full-valence active space is huge in our case.
The orbital entanglement diagrams confirm that such an extension of the valence space is justified and most likely needed for a reliable multi-configurational description.

While all iron-sulfur clusters considered are clear class-2 cases in a standard (delocalized) orbital basis, properly localized orbitals can shrink the number of strongly entangled, hence active orbitals to a size typical for class-1 molecules. However, for large clusters such as FeMoco, this is no longer the case and its electronic structure remains that of class-2 for any orbital set. This is also well reflected in the structure of the threshold diagrams (see Fig.\ \ref{fig:fes-comp} for a comparison).
Hence, a sufficiently large number of (bridged) open-shell transition metal atoms produces a class-2 situation. Other examples for which this can be expected, are the 
active site of FeFe hydrogenase \cite{Finkelmann2014},
the P-cluster of nitrogenase \cite{Li2019b},
and possibly the manganese cluster in the photosystem II \cite{Naja2014}
to mention only a few.

\begin{figure}[htp!]
    \centering
    \includegraphics[width=1.0\linewidth]{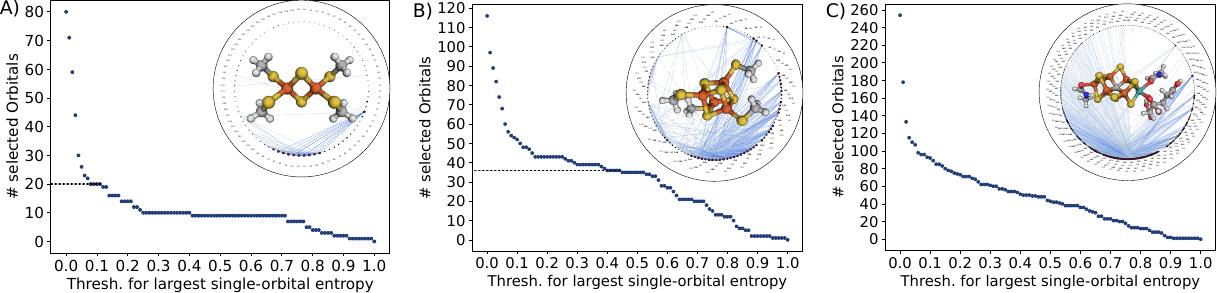}
    \caption{
Orbital entanglement diagrams (inlays) and threshold diagrams for [2Fe-2S] (A), [4Fe-4S] (B) and FeMoco (C) in a split-localized natural orbital basis (set 4 described in the appendix, 
which resembles the orbitals in Ref.\ \citenum{lee2023}). Color code: yellow denotes sulfur atoms, orange iron atoms , red oxygen atoms, gray carbon atoms, dark blue nitrogen atoms, white hydrogen atoms, and turquoise molybdenum.
A) Full-valence DMRG ($m=250$) diagrams for an initial active space of 80 orbitals and 82 electrons.
The horizontal dotted line in the threshold diagram shows a plateau of 20 orbitals and 30 electrons, which corresponds to the manually chosen active space in Ref.\ \citenum{lee2023}.
B) Large-CAS (30 orbitals per sub-CAS) for an initial active space consisted of 116 orbitals and 110 electrons.
The horizontal dotted line in the threshold diagrams shows the manually selected active space CAS(52, 36) from Ref.\ \citenum{lee2023}.
The two plateaus above this active space result in a CAS(62, 43) and CAS(56, 39), while the one directly below yields a CAS(50, 35).
C) Large-CAS (30 orbitals per sub-CAS) for an initial active space of 249 electrons in 254 orbitals.
    }
    \label{fig:fes-comp}
\end{figure}

To demonstrate that the analysis generalizes to other types of molecules, we inspect buckminster fullerene C$_{60}$ with 240 electrons in 240 orbitals.
Due to the large conjugated $\pi$-system in C$_{60}$, a final active space consisting of 60 orbitals can be selected, which is in line with chemical intuition. However, an active space with 60 electrons and orbitals is already challenging to converge. It is obvious that class-2 molecules can be easily generated in this way by simply enlarging the number of $sp^2$-carbon atoms in the conjugated structure (as in fullerenes larger than C$_{60}$, in annelated (hetero)aromatic ring systems, in graphene flakes, or in unsaturated organic macromolecules).

\begin{figure}[htp!]
    \centering
    \includegraphics[width=0.45\linewidth]{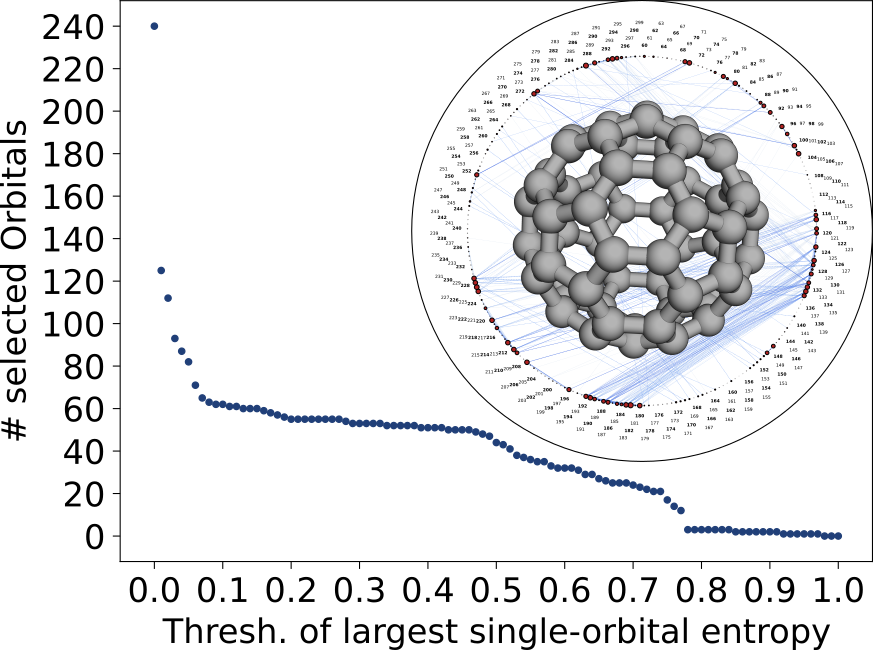}
    \caption{Threshold and orbital entanglement (inlay) diagrams for the fullerene C$_{60}$ in $I_h$ symmetry (structure also shown in ball-and-stick representation) based on split-localized natural orbitals. 
    }
    \label{fig:c60}
\end{figure}

Initial state preparation will be challenging for class-2 molecules.
For FeS clusters, it is clearly advantageous to switch from a determinantal to a configuration-state-function basis (see the discussion in the appendix). Still, depending on the size and bonding of a cluster, millions of basis functions might be required to achieve a significant overlap with the target state. Although it would be feasible to initialize millions of determinants (compared to the effort required for QPE), it can be hard to sample millions of basis states from an MPS. Hence, it would be most advantageous to initialize the approximate MPS directly that has been obtained while selecting the active orbital space. Algorithms to prepare an MPS on a quantum computer have already been proposed \cite{2023MalzMPS,Smith2024,Iaconis2024,Smith2023} (see also the appendix).
One might argue that the approximate MPS could be qualitatively wrong (for instance, a bond dimension of $m$=250 applied here for FeS clusters might be considered to be an order of magnitude to small in order to be reliable), but this is in fact not an obstacle.
While the quality of the approximate MPS is indeed unknown, quantum computation can help investigate its accuracy. If the approximate MPS is a good trial state, we will obtain the ground state with QPE. Otherwise, QPE can help guide us towards a better wave function.

\section{Conclusions}

In this work, we 
defined three classes of electron correlation problems based on the distribution of single-orbital entanglement entropies (arranged in threshold diagrams). 

Class-0 molecules are typical single-reference cases where the Hartree--Fock determinant dominates by far the determinantal expansion of the target state. The electronic structures of these molecules are dominated by dynamic electron correlation, which is well captured by coupled cluster models in traditional computations. However, we emphasize that no accuracy estimates are available for individual molecules when applying these models so that error-controlled QPE algorithms possess a clear advantage over traditional quantum chemical calculations. 

By contrast,
class-1 molecules represent the standard multi-configurational problem in ground-state (reaction) chemistry, which can still be initialized by a single Slater determinant or, in order to push the overlap with the target state to more than 90\%, by a small number of Slater determinants, as demonstrated in this work on a set of prototypical cases. Determinants to be initialized can be easily sampled from an approximate MPS (which will be available if the AutoCAS algorithm is used to determine the active orbital space) or taken from an approximate CI singles-doubles calculation as in multi-reference CI models.

The truly challenging cases for guiding-state preparation are in 
class 2. These are found for electronically excited states and for special types of molecules that show truly strong electron correlation in the ground state. The prime example of the latter are open-shell $3d$ transition-metal clusters. In such cases, a reduction of the full valence orbital space will hardly be possible resulting in a huge number of determinants to be initialized as a superposition. However, even then, suitable guiding states will still accessible from approximate full-valence CAS considerations, as we have demonstrated at the example of FeMoco (although this will require additional work to merge approximate MPS results of large-CAS calculations from the AutoCAS protocol).

The three classes are relevant for (finite-sized) molecules as they occur in the atomistic description of chemical reactions (which requires hardly more than about hundred atoms). As such, reactive systems (i.e., stable intermediates and transition states), which are typically found in classes 0 or 1,  are not affected by the orthogonality catastrophe and guiding state preparation will easily be possible. Already the Hartree--Fock determinant features a sufficiently large overlap with the target state, even for class-1 molecules. This finding extends to embedding descriptions that make accurate quantum simulations accessible for large systems as they occur in materials science or biochemistry (see, for instance, Ref.\ \citenum{Erakovic2024}).

The main conclusion from this work is that computational reaction chemistry on reactive systems containing up to a few hundred atoms will be amenable to quantum computation and will not be plagued by the orthogonality catastrophe in the guiding-state initialization. We emphasize that these conclusions will also hold in a broader context where finite subsystems can be used to model some atomistic problem (such as small elementary cells in materials science or solid state physics). 

The orbital spaces of class-1 molecules will usually be small so that traditional quantum chemical multi-configurational approaches can be applied.
Hence, one might argue that these molecules will not be those for which a quantum advantage can be demonstrated.
While class-2 molecules have been in the focus as key targets for quantum computation, class-1 molecules are the ones typically found in routine reaction chemistry and point to the need for utility-scale quantum computation in the long run.
However, the identification of class-2 molecules (based on a screening of approximate MPSs) will show the way to most critical electronic structures for which a quantum advantage may be demonstrated first.
Threshold diagrams provide a convenient way to distinguish between the two classes, following the qualitative discussion of their different patterns presented in this work. Although there may be boundary cases, it can be expected that they are comparatively rare for local electronic structures of chemical reactants.

In the light of future routine fault-tolerant quantum computations on sufficiently large orbital bases, we note that those orbitals that are neglected in an active space will make a nonnegligible contribution to the wave function and electronic energy.
These dynamic correlations are important for quantitative accuracy and cannot be ignored.
It is necessary to account for them; for instance, by multireference perturbation theory in a QPE framework\cite{gunther2024}. 
Although this is a viable strategy for future intermediate-size fault-tolerant quantum computation, a true breakthrough will be achieved when larger machines will become available that would allow one to consider active spaces of on the order of 1000 or more spatial orbitals.
Then, almost all reaction chemistry could be treated with error control through QPE, no additional dynamic correlations would need to be taken into account and state initialization could still be done based on the small-CAS information (obtained, for instance, from an approximate DMRG calculation as done in this work).

\section*{Acknowledgments}
M.R. gratefully acknowledges generous financial support by the Swiss National Science Foundation (grant no. 200021\_219616), the Novo Nordisk Foundation (grant no. NNF20OC0059939 ‘Quantum for Life’), and through ETH Research Grant ETH-43 20-2.
Part of this work was presented by M.R. at the ACS Fall Meeting in Denver in August 2024.

\begin{appendix}
\section{Computational Approach}

To analyze the weight of the Hartree--Fock determinant in the determinantal expansion of target states, we chose a set of molecules that belong to classes 0, 1, and 2.
Most of these molecules have been known to exhibit a multi-configurational character of some degree:
C$_4$H$_4$\cite{lyakh2011},
O$_3$\cite{pulay1988},
$^1$C$_6$H$_4$\cite{evangelista2012}, 
$^3$C$_6$H$_4$\cite{evangelista2012}, 
N$_2$O$_4$\cite{pulay1988},
$^1[$Fe(NO)(CO)$_3]^-$\cite{sayfutyarova2017} 
and $^3[$Fe(NO)(CO)$_3]^-$\cite{sayfutyarova2017}
(we took the Cartesian coordinates of the compounds from the respective references).
We expect all of these molecules to be in class 1 and added the water molecule as a typical class-0 molecule for comparison (the structure of the water molecule was taken from Ref.\ \citenum{Hoy1979}).
Note that all these structures are equilibrium structures, except for cyclobutadiene C$_4$H$_4$, which is a transition state structure.

As potential candidates for molecules in class 2, we selected an iron-porphyrin model complex, abbreviated as Fe-Porph in the following, of the active site of a heme enzyme (myoglobin) with a coordinated glycine and an imidazole residue (Cartesian coordinates were taken from Ref.~\citenum{Tinzl2024}). 
Moreover, we investigated the fullerene C$_{60}$\cite{Sure2017} in $I_h$ symmetry and three iron-sulfur clusters, namely 
[Fe$_2$S$_2$(SCH$_3$)$_4$]$^{2-}$\cite{sharma2014} denoted as [2Fe-2S],
[Fe$_4$S$_4$(SCH$_3$)$_4$]$^{2-}$\cite{sharma2014} denoted as [4Fe-4S],
and the FeMo-cofactor (FeMoco) of nitrogenase with anchoring amino acid residues (as shown in Fig.\ \ref{fig:femoco} below).

We optimized 
the molecular structure of [2Fe-2S] taken from Ref.\ \citenum{sharma2014} with broken-symmetry density functional theory (BS-DFT)\cite{Noodleman1981} implemented in the ORCA program package\cite{neese2022} with the PBE\cite{Perdew1996} functional, D3BJ dispersion correction\cite{Grimme2011}, the def2-SVP\cite{weigend2005} basis set, and an initial spin multiplicity of 11.
After convergence of this high-spin state, the local spin is flipped on one metal center.
Afterwards, the Kohn-Sham orbitals for a singlet state were optimized in an unrestricted calculation (with the same number of $\alpha$- and $\beta$- spin orbitals to model an open-shell singlet state with antiferromagnetic coupling).
This procedure is a standard procedure for a DFT description of a multi-center cluster, since both Fe$^{3+}$ atoms are in their high-spin configuration with a local spin quantum number of $5/2$.

We took the molecular structure of [4Fe-4S] from Ref.\ \citenum{sharma2014} and optimized the low-spin broken-symmetry determinant similar to the procedure described for [2Fe-2S]. 
However, since this cluster now contains four iron centers, the initial spin multiplicity is 21, while the final state is again a singlet. 

The Cartesian coordinates of the FeMoco were extracted from the X-ray structure 3U7Q\cite{spatzal2011} downloaded from the protein database. 
Since the FeMoco can be extracted from the protein in N-methylformamide (NMF) in experiment\cite{Shah1977}, we only retained the FeMoco with a bound homocitrate ligand and two NMF solvent molecules that saturate the coordination sphere. 
This structure was then optimized with BS-DFT (PBE functional\cite{Perdew1996}, def2-SVP basis set\cite{weigend2005}, effective core potential for molybdenum \cite{andrae1990}, and
D3BJ dispersion corrections\cite{Grimme2011}) with Turbomole\cite{Ahlrichs1989, Franzke2023} in an analogous procedure as described above for [2Fe-2S].
The resulting electronic structure corresponds to the experimentally known $M_S=S=3/2$ quartet resting state.\cite{Munck1975, Rawlings1978} 
We verified that the local-spin coupling scheme corresponds to the energetically most stable ``BS-7'' spin alignment which is the same as in Refs.~\citenum{Lukoyanov2007, Kastner2007, Harris2011, Dance2011, Bjornsson2017}.

Since the orbital basis can heavily affect the multi-configurational character of the wave function, different orbital choices can severely change the determinantal composition of the ground state (see also \cref{sec:orbitals} and \cref{sec:proof}), with consequences for guiding-state preparation.
For the iron-sulfur clusters, the BS-DFT optimization of the desired total spin state for the resting state can be difficult because of the large number of orbitals with similar energy and the low-lying excited states of different spin configuration.
We accomplished convergence of all FeS clusters by a sequence of level shifts of decreasing size (the final optimization step was then accomplished without a level shift).

Moreover, we decided on a second strategy to obtain another set of molecular orbitals that can then be compared to the broken-symmetry solution.
To obtain a second set of orbitals, we followed the work of Chan and co-workers\cite{Li2019,lee2023} and optimized the high-spin state, which can be represented by a single determinant with 35 unpaired electrons (spin multiplicity is 36).
In addition, we explore split-localized natural orbitals from both high- and low-spin unrestricted DFT solutions (see section \ref{sec:orbitals}).
Note also that we have not considered scalar-relativistic effects for all atoms but molybdenum as these are small for iron and the main-group elements.

For the sake of completeness and reproducibility, we collected the coordinates of all compounds and provided them at the end of this supporting information.

As an approximation to the target states to which we can compare initial state information, we optimized matrix product states (MPSs) with the density matrix renormalization group (DMRG) algorithm\cite{White92, White93} implemented in the program QCMaquis\cite{keller2015a,qcmaquis_zenodo}.
For all DMRG single-point calculations 
on all molecules but Fe-Porph, C$_{60}$, and the Fe-S clusters, we optimized Hartree--Fock orbitals in a cc-pVDZ\cite{dunning1989,balabanov2005,balabanov2006} basis set with pySCF\cite{pyscf}. pySCF is interfaced to our DMRG program QCMaquis\cite{keller2015a,qcmaquis_zenodo} and provides the one- and two-electron integrals in the molecular orbital basis for the second-quantized electronic Hamiltonian on the fly. 

For Fe-Porph, the Fe-S clusters and the buckminsterfullerene C$_{60}$, unrestricted Kohn-Sham DFT single-point calculations delivered PBE\cite{Perdew1996}/def2-SVP\cite{weigend2005} orbitals (taking the effective core potential for Mo in FeMoco into account) in pySCF\cite{pyscf} following the same orbital preparation protocol as used for the structure optimization (see above), in order to produce the integrals in the molecular orbital basis for QCMaquis.

Single-orbital entropies were taken as a measure for the importance of an orbital in the active space \cite{Boguslawski2012,Stein2016}.
For the initial approximate DMRG calculations (choosing a low bond dimension of $m=250$ and five forward-and-backward sweeps, if not noted otherwise), all valence orbitals were selected for an \textit{initial} active space.
The resulting MPSs served to evaluate approximate single-orbital entropies\cite{Legeza2003}.
In addition, we provide results on mutual information measures \cite{Legeza2006,Rissler2006} in orbital entanglement diagrams that show orbital entropies for orbitals arranged on a circle with mutual information encoded in the thickness of the connecting lines\cite{Freitag2015}. 
The \textit{final} active spaces, for which DMRG calculations were then converged to obtain a reliable representation of the target states, consisted solely of the strongly entangled orbitals as identified by this orbital selection scheme \textsc{AutoCAS}\cite{Stein2016,Stein2019,autocasZenodo}.
\textsc{AutoCAS} allows one to run multi-configurational calculations in a black-box fashion (see, for instance, its cloud-based implementation in the framework of Azure Quantum Elements\cite{Unsleber2023AutoRXN}).

Because the valence orbital spaces of [4Fe-4S], Fe-Porph, C$_{60}$, and FeMoco consisted of more than 100 electrons and orbitals (including the $4d$ double-shell orbitals for all iron centers to account for double-$d$-shell correlation effects\cite{Andersson_1992_doubledshell, Pierloot_1995_doubledshell}), we applied the large active space protocol of the \textsc{AutoCAS} program\cite{Stein2019} for active space selection. This protocol partitions the large valence active space into smaller spaces from which the most strongly entangled orbitals can then be chosen. For FeMoco, we note that the number of electrons to be considered in this initial screening was 249 and the corresponding number of orbitals 254.

Since the large active space protocol in Ref.\ \citenum{Stein2019} was described for closed-shell molecules, we generalized it to open-shell systems. 
To conserve the spin-multiplicity in the smaller spaces, the energetically highest, formerly singly occupied orbitals, were included in each smaller space. 
As in the original implementation, each smaller space is evaluated by DMRG and the corresponding single-orbital entropies and mutual information measures are recombined based on their maximum value  
in order to approximate the initial large (i.e., full-valence) orbital space.
Subsequently, the active space is selected by the plateau-based algorithm for threshold diagrams \citenum{Stein2016}.

To identify the important determinants in an MPS, we evaluated the weight $|c_i|^2$ of a determinant by explicit contraction of the MPS\cite{Gerrit2007MPStoCI} in the stochastic sampling procedure SR-CAS \cite{Boguslawski2011} with determinant $i$ to obtain its configuration interaction (CI) coefficient $c_i$.
For the converged MPS, we calculated the accumulated weight $A_w$ of the important determinants, 
\begin{align}
    A_w = \sum_{i=1}^{N_{d}} |c_i|^2
\end{align}
of the $N_d$ most important determinants as revealed by the SR-CAS algorithm (with a threshold of the absolute value of the configuration interaction coefficients of 0.00001).

We also evaluated the MPS overlap of the approximate MPS with the converged MPS for the \textit{final} active space of each molecular structure.
This required another approximate DMRG calculation (with low bond dimension and low number of sweeps) for the selected orbitals before convergence is then reached for large bond dimension.
Note that the additional approximate DMRG calculation is always fast after the orbital selection step, but the converged DMRG calculation may not be accessible, which is the reason for turning to a quantum computer.
The MPS overlap measures calculated here shall demonstrate the high fidelity of an approximate MPS that is comparatively easy to obtain and can therefore serve in the guiding-state preparation step.

\section{Threshold diagrams}

In threshold diagrams, the $x$-axis lists single-orbital entropies taken as thresholds to count how many orbitals are assigned this or a larger value. This number of orbitals is then given by the $y$-axis.
Hence, every point in such a diagram denotes how many orbitals have a factional single-orbital entropy larger than the corresponding $x$-value.
Single-orbital entropies\cite{Legeza2003} are given by
\begin{equation}
  s_i(1) = - \sum_{\beta = 1}^4 w_{\beta, i} \ln w_{\beta, i}
  \label{eq:singleOrbitalEntropy}
\end{equation}
where $w_{\beta, i}$ are the eigenvalues of the one-orbital reduced density matrix for orbital $i$ and $\beta$ is one of the four possible configurations for a spatial orbital (i.e., unoccupied, $\alpha$- or $\beta$-occupied, or doubly occupied).
The mutual information\cite{Rissler2006, Legeza2006} of two orbitals $i$ and $j$
\begin{equation}
  I_{ij} = \frac{1}{2} \left[ s_i(1) + s_j(1) - s_{ij}(2) \right] (1 - \delta_{ij})
  \label{eq:mutualInformation}
\end{equation}
requires the corresponding two-orbital entropy
\begin{equation}
  s_{ij}(2) = - \sum_{\beta = 1}^{16} w_{\beta, ij} \ln w_{\beta, ij},
  \label{eq:twoOrbitalEntropy}
\end{equation}
where $w_{\beta, ij}$ are the eigenvalues of the two-orbital reduced density matrix and $\beta$ now refers to all 16 occupations possible for two spatial orbitals.

After reordering and normalizing (by the largest element) the
single-orbital entropies according to their value, the AutoCAS 
algorithm 
is searching for plateaus,
which are plotted in a threshold diagram. 
All orbitals that belong to a plateau
have at least a single-orbital entropy, which is larger than
a threshold percentage with respect to the largest element. 
As weakly correlated orbitals only have small single-orbital
entropies, these orbitals drop out at low-threshold percentages. Hence, these plateaus
are employed in the algorithm to identify strongly correlated orbitals. Another benefit
of the plateaus is that unexpected, but strongly correlated orbitals are found and
included in the final active space calculation. 

Fig.\ \ref{fig:ideal_thresholds} A) shows an idealized threshold diagram of a class-0 molecule. Since class-0 molecules are not strongly correlated, no entropy is larger than a threshold (10~{\%} of the theoretical maximum) and hence they are single-configurational. This is expressed by the steep drop in the orbital number at a low threshold.
Fig.\ \ref{fig:ideal_thresholds} B) shows an idealized threshold diagram of a class-1 molecule. Due to the moderate correlation of some orbitals in class-1 molecules, a typical plateau structure emerges. The initial drop in the number of orbitals results from the dynamically correlated orbitals in the initial active space.
After the screening of the dynamically correlated orbitals, some plateaus emerge, which indicate orbitals with a similar degree in correlation. 
The system dependent-scaling in the diagram shows the same plateau behavior without the tail at 0 selected orbitals, since usually the maximum single orbital entropy is well below the theoretical maximum.
Fig.\ \ref{fig:ideal_thresholds} C) shows an idealized threshold diagrams of class-2 molecules. As these molecules exhibit strong correlation between many electrons, no clear plateau structure emerges. Since many electrons are correlated, with a slightly different single orbital entropy, a small number of orbitals drops out even after small changes in the threshold percentage. Additionally, due to strong correlation some single-orbital entropies are close to the theoretical maximum so that the system dependent scaling cannot increase the visibility of plateaus. 

\begin{figure}
    \centering
    \includegraphics[width=\textwidth]{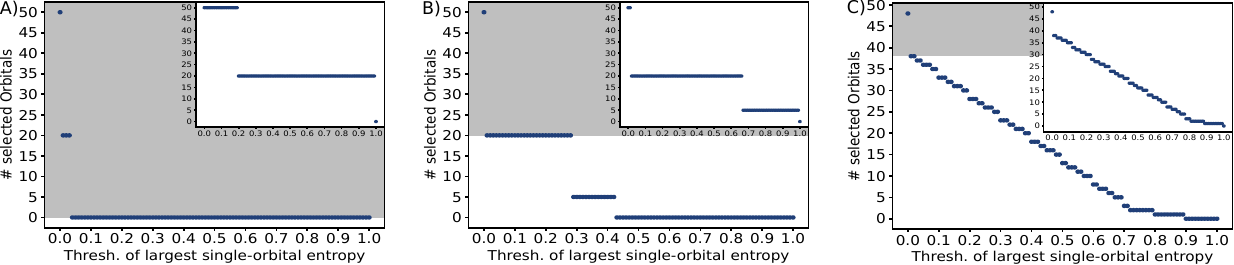}
    \caption{Idealized threshold diagrams for different correlation classes are shown.
    The larger diagram shows the thresholds normalized by the theoretical maximum.
    The grey area indicates dynamically correlated orbitals.
    In the inlay diagram the thresholds are normalized by the maximum entropy of the system, as these are interpreted by AutoCAS.
    A) class-0 molecules, B) class-1 molecules, and C) class-2 molecules.
    }
    \label{fig:ideal_thresholds}
\end{figure}

\section{Additional Discussion}

\subsection{Class-1 Cases}

All class-1 molecules show a typical plateau structure in the threshold diagrams defined in Ref.\ \citenum{Stein2016} (see Fig.\ \ref{fig:class_1_threshold}), 
which is the basis for the reliable selection of orbitals by the AutoCAS algorithm. 
Note that the shown threshold diagrams are scaled by the maximum entropy of the systems, if not mentioned otherwise.

\begin{figure}[htp!]
    \centering
    \includegraphics[scale=0.8]{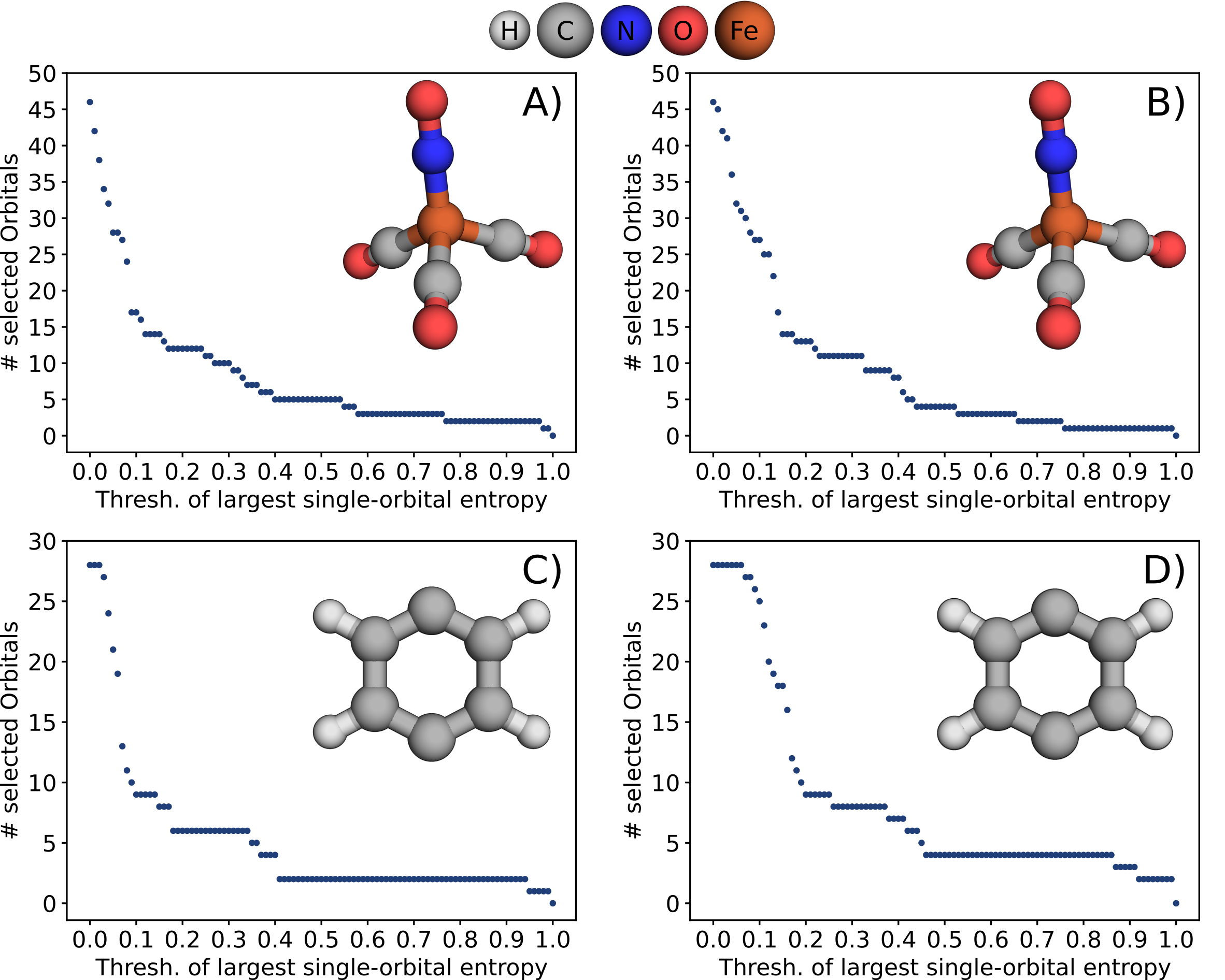}
    \caption{Threshold diagrams (as defined in Ref.\ \citenum{Stein2016}) 
    for the four class-1 molecules A) $^1[$Fe(NO)(CO)$_3]^-$, B) $^3[$Fe(NO)(CO)$_3]^-$, C) $^1$C$_6$H$_4$ and D) $^3$C$_6$H$_4$. On the $x$-axis, the entropy of an orbital divided by the maximum entropy found for some orbital in system is given. The $y$-axis lists the number of orbitals with an entropy larger than the threshold defined by the value on the $x$-axis.}
    \label{fig:class_1_threshold}
\end{figure}

\subsection{Orbital Effects\label{sec:orbitals}}

We investigated the effects of different orbital choices (see Table \ref{tab:orbital_choices} for an overview) in the large CAS protocol. 
In view of the approximate nature of this protocol, we investigated the convergence with the number of orbitals per sub-CAS for each orbital set.

\begin{table}[htp!]
    \centering
    \begin{tabular}{c l}
    \hline\hline
    Set 1: & BS-DFT orbitals \\
    Set 2: & BS-DFT, split-localized natural orbitals \\
    Set 3: & High-spin UKS orbitals \\
    Set 4: & High-spin UKS, split-localized natural orbitals  \\
    Set 5: & High-spin UKS, full-split-localized natural orbitals \\
    \hline\hline
    \end{tabular}
    \caption{Overview of different orbital sets employed in DMRG calculations on [2Fe-2S] and [4Fe-4S].}
    \label{tab:orbital_choices}
\end{table}

Since the authors of Ref.\ \citenum{lee2023} localized the orbitals, we also investigated the effect of orbital localization on the active space selection.
To keep the algorithm as generally applicable as possible, we separately localized different orbital spaces, which we select based on natural occupation numbers (NON) 
(see Fig.\ \ref{fig:orbital_choices}). The first localization space consists of all orbitals with $\text{NON} > 1.98$ and the second space of orbitals with $1.98 \geq \text{NON} \geq 0.02$ (orbital sets 2 and 4 in Table \label{tab:orbital_choices}). 
We applied the thresholds as in Ref.\ \citenum{Keller2015b}, where the second localization space corresponds to the unrestricted natural orbital selected CAS. 
Since the initial active space in AutoCAS is generally larger, we additionally localize the remaining orbitals with $\text{NON} < 0.02$ (orbital set 5 in Table \label{tab:orbital_choices}). 

\begin{figure}[htp!]
    \centering
    \includegraphics[width=.6\linewidth]{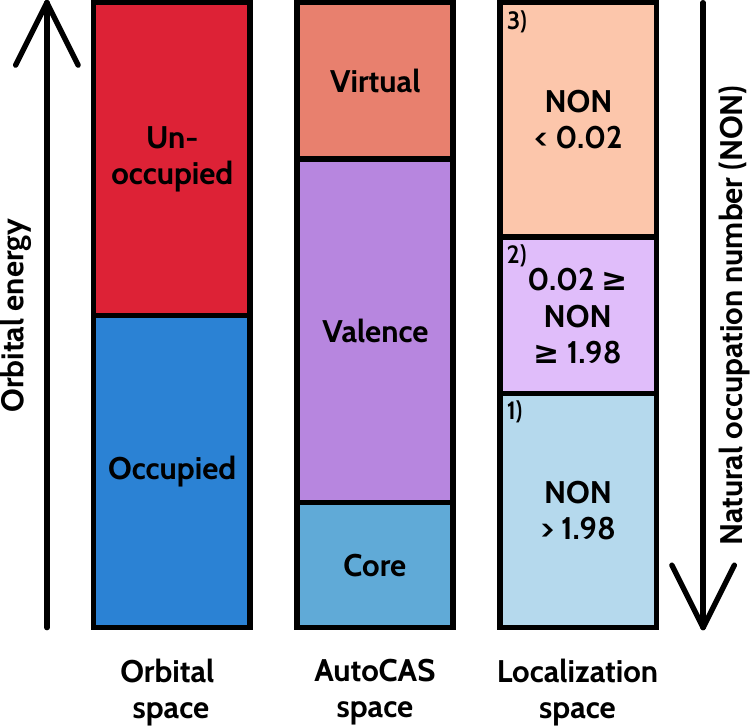}
    \caption{Three different sub-divisions of orbitals. 
    The orbital space of the Hartree--Fock determinant may be divided into \textit{occupied} and \textit{unoccupied} orbitals.
    In the AutoCAS algorithm, the initial orbital space is divided into a \textit{core}, \textit{valence}, and \textit{virtual} space. Since AutoCAS aims to function in a black-box manner, the valence space is determined by selecting the number of valence orbitals around the Fermi-Vacuum. 
    For the split localization, the orbital space is divided into three-spaces based on natural occupation numbers according to thresholds of 1.98 and 0.02 (see Ref.\ \citenum{Keller2015b}).
    }
    \label{fig:orbital_choices}
\end{figure}

We analyzed the active space selection of [2Fe-2S] and [4Fe-4S] clusters for different orbital choices. The standard approach (of questionable reliability) in computational chemistry to describe these clusters is BS-DFT which emulates the anti-ferromagnetic coupling by an unrestricted determinant and evaluates the corresponding energy with a standard exchange-correlation density functional approximation that is not designed for such spin states \cite{Jacob2012}.

The [2Fe-2S] cluster is the perfect test case to investigate the orbital choice, as both iron centers are antiferromagnetically coupled into a singlet but locally in a high spin configuration. 
The valence space of the [2Fe-2S] cluster, CAS(82, 80), is small enough to be evaluated with DMRG directly. It does not require the large-CAS protocol.
We first optimized the high-spin determinant (orbital set 3) and subsequently flipped the local spin on one iron atom, to converge the broken-symmetry, low-spin determinant (orbital set 1). 
The natural orbitals of set 1 and set 3 were split-localized in localization space 1) and 2) (see Fig.\ \ref{fig:orbital_choices}) to generate orbital sets 2 and 4, respectively.  
Orbital set 5 was generated by split-localization of the natural orbitals of set 1 in the three localization spaces 1), 2), and 3) shown in Fig.\ \ref{fig:orbital_choices}.

We selected active spaces for all 5 orbital sets with the default AutoCAS algorithm and compared the results with those obtained for the large active space protocol with 10, 20, and 30 orbitals per sub-CAS. In this way, we could assess the convergence and thereby the reliability of the large-CAS protocol, which we afterwards applied to [4Fe-4S] and FeMoco where no all-valence AutoCAS selection is possible in one shot. 

For the full-valence DMRG($m=250$) calculation, 
the AutoCAS suggested active space for every orbital set consisted of 10 to 12 orbitals.
By contrast, the orbital choice had a dramatic effect of the CAS selection based on the large-CAS protocol.
Orbital sets 1, 2, and 3 all suggested relatively large active spaces for 10 and 20 orbitals per sub-CAS. For 30 orbitals per sub-CAS, no plateau could be determined for the delocalized orbital sets (sets 1 and 2). 
For set 4, AutoCAS suggested an active space between 13 and 19 orbitals for any number of orbitals per sub-CAS. 
The reason for this behavior is that 
there are too many different single-orbital entropies so that no large plateau can form, 
because there is no clear distinction between strongly  and weakly correlated orbitals in these systems.
As a result, these active spaces contained $3d$-type orbitals of the Fe centers, but no orbitals living on the sulfur bridges.
Since the maximal single-orbital entropy for [2Fe-2S] is close the theoretical maximum, the length of a plateau shortens so that plateaus cannot be identified when following the default plateau threshold in the AutoCAS algorithm. 
After decreasing the minimum plateau length, AutoCAS suggested reasonable active spaces with 16 orbitals, 14 orbitals, 17 orbitals, and 20 orbitals for set 1, set 2, set 3, and set 4, respectively, now also including bridging sulfur-including orbitals.

\begin{figure}[htp!]
    \centering
    \includegraphics[width=.65\linewidth]{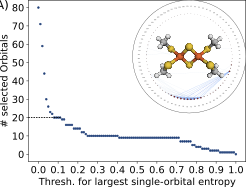}
    \caption{Full-valence DMRG ($m=250$)  
    based entanglement and threshold diagram of the [2Fe-2S] cluster, with split-localized natural orbitals from the high-spin determinant (for orbital set 4, which resembles the orbitals if Ref.\ \citenum{lee2023}).
    The initial active space consisted of 80 orbitals and 82 electrons. 
    The horizontal dotted line in the threshold diagram shows a plateau of 20 orbitals and 30 electrons, which corresponds to the manually chosen active space in Ref.\ \citenum{lee2023}.
    }
    \label{fig:fe2s2}
\end{figure}

The larger iron-sulfur cluster [4Fe-4S] features an
initial active space of 110 electrons and 116 orbitals. A full-valence DMRG calculation was therefore hardly feasible.
Hence, we optimized the orbital sets 1, 3, 4, and 5 for this system and searched for the active with the large-CAS protocol with 30 orbitals per sub-CAS. 

Similar to the smaller [2Fe-2S] complex, the delocalized orbitals (sets 1 and 3) showed much stronger (approximate) entanglement, by contrast to sets 4 and 5. 
For orbital set 4, multiple small plateaus formed, starting with 36 orbitals. For the split localized orbital, plateau formation occurred at 43 orbitals. 
The difference in the number of orbitals is due to the approximate treatment of the active space in the large-CAS protocol.
However, since the large-CAS protocol results usually in too large active spaces, we can use it to reduce the initial active space size and evaluate the smaller active space with a single subsequent DMRG calculation.
Moreover, the smaller plateaus are likely to deliver the correct answer (basically this is the point where the entropies show  a kink, usually at low threshold percentages).

\begin{figure}[htp!]
    \centering
    \includegraphics[width=.65\linewidth]{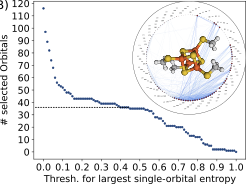}
    \caption{Large-CAS (30 orbitals per sub-CAS) entanglement (inlay) and threshold diagrams of the [4Fe-4S] cluster with split-localized natural orbitals obtained from the high-spin determinant (set 4).
    The initial active space consisted of 116 orbitals and 110 electrons. 
    The horizontal dotted line in the threshold diagrams shows the manually selected active space CAS(52, 36) from Ref.\ \citenum{lee2023}.
    The two plateaus above this active space result in a CAS(62, 43) and CAS(56, 39), while the one directly below yields a CAS(50, 35).
    }
    \label{fig:fe4s4}
\end{figure}

In view of the CAS selections found for the smaller iron-sulfur cluster [2Fe-2S], orbital localization is key to reduce orbital entanglement and eventually allows one to reduce the active space size. 
Hence, these iron-sulfur clusters are clear class-2 cases in a delocalized orbital basis that does not encode the qualitative features of their electronic structures well. By contrast, properly localized orbitals shrink the number of strongly entangled, hence active orbitals to a size typical for class-1 molecules. This is also well reflected in the structure of the threshold diagrams.

For [2Fe-2S], the plateau corresponding to the manually selected active space of Ref.\ \citenum{lee2023} shows a width of only four threshold steps.
The next plateau at 10 orbitals would lead to the inclusion of only the $d$-orbitals centered at the iron atoms, but it would miss orbitals on the bridging sulfur atoms. These edge cases also show that the original plateau threshold value of 10 threshold steps defined in Ref.\ \citenum{Stein2016}, which is the basis for the active orbital selection in AutoCAS, is too large to consistently select active spaces for class-2 molecules; hence, it would need to be modified for them.

\subsection{Class-2 Case: FeMoco}

First of all, we note that the valence spaces to be considered in the orbital selection procedure of class-2 cases are in general very large. 
For our FeMoco structural model we even have 249 electrons in 254 orbitals (CAS(249, 254)).
Although large orbital spaces do not necessarily produce strong electron correlation, this may be expected here considering the number of unsaturated bonds and the open-shell metal atoms.
In the following, we discuss the results of the AutoCAS selection based on approximate DMRG ($m=250$, five sweeps) results, 
which will make a clear point for the assignment of these molecules to class 2.
Fig.\ \ref{fig:femoco} shows the threshold and orbital entanglement diagrams for FeMoco.
It is obvious from these diagrams that both possess strongly correlated electronic structures indicated by high single-orbital entropies and mutual information values.

It was already recognized by Chan and co-workers \cite{Li2019} that a full valence space for FeMoco is likely to be necessary to represent the S=3/2 resting state of FeMoco.
However, these authors proposed a significantly smaller active space of 113 electrons in 76 spatial orbitals (including Fe $3d$, S $3p$, Mo $4d$, C $2s2p$, and ligand orbitals) than what we consider in this paper.
The reason for this large discrepancy is likely the fact that we include the double-$d$ shell for all metal centers and also the valence-shell contributions of all atoms in our FeMoco model. Also note that we have different ligands coordinating, which is, however, not decisive for the fact that the full-valence active space is huge in our case.
The orbital entanglement diagrams confirm that such an extension of the valence space is justified and most likely needed for a reliable multi-configurational description.

As in the case of the [2Fe-2S] and [4Fe-4S] clusters, we found 
a qualitative difference in the orbital entanglement diagrams of the FeMoco resting state represented by broken-symmetry orbitals (set 1) versus the specifically prepared all-high-spin orbitals (set 5). The set-1 threshold diagram delivered no conclusive result on the choice of the active space, instead all valence orbitals have to be considered yielding an active space of 167 electrons in 208 orbitals. By contrast, the set-5 threshold diagram still shows a class-2 case, but the active space can be reduced to 117 electrons in 104 orbitals by virtue of this specific choice of orbitals.
Hence, an efficient reduction in CAS size can be expected by preferring localized orbitals according to set 5 over delocalized orbitals from BS-DFT (set 1). 

As a side remark, we note that a further reduction of the active space might be accomplished by inspecting the second entanglement measure, i.e., the mutual information, which is not required for the standard AutoCAS algorithm that relies on single-orbital entropy information only (cf. the reduced mutual information values in Fig.\ \ref{fig:femoco}).

The large-CAS algorithm of Ref.\ \citenum{Stein2019} can be applied for class-2 molecules with benefit as it allows one to identify a class-2 case (the orbital entanglement diagram will show all orbitals to take part in the entanglement pattern and the threshold diagrams will point to class 2).

\begin{figure}[htp!]
    \centering
    \includegraphics[width=0.65\linewidth]{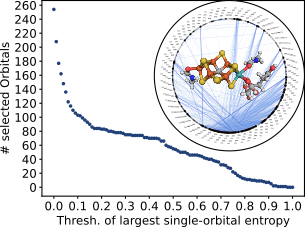}
    \caption{Threshold and orbital entanglement (inlay) diagrams for the FeMoco structure model (also shown in ball and stick representation) based on high-spin split-localized natural orbitals (set 5).
    Each color of a ball represents an atom, namely: yellow corresponds to sulfur, orange corresponds to iron, red corresponds to oxygen, grey corresponds to carbon, dark blue corresponds to nitrogen, white corresponds to hydrogen and turquois corresponds to molybdenum.
    }
    \label{fig:femoco}
\end{figure}

\subsection{Preparation of spin-coupled determinants}
It is evident that the active spaces required for the class-2 molecules discussed above are huge and, accordingly, the number of determinants that contribute to the wave function is enormous. However, many determinants emerge from the coupling of local spins in the case of open-shell transition metal clusters (by simply flipping spin in the optimized orbital basis). This special type of correlation has very recently been highlighted for quantum computing applications \cite{martidafcik2024spin},
and it is fully in line with our general understanding of electron correlation in antiferromagnetically coupled clusters, which has recently also been emphasized in the context of traditional quantum chemistry\cite{Limanni2020,Limanni2021,Dobrautz2021,LiManni_2023,Izsak2023}.
We therefore consider next how much could be gained by exploiting linear combinations of determinants that are properly spin coupled, i.e., configuration state functions (CSFs). 

Before we delve into this option, we already note that exploiting this idea in practice for state preparation would require well localized orbital bases that would allow one to identify
alpha-spin orbitals localized on individual metal centers in the molecular orbital basis. In other words, the regular case for large molecules will be a mixture of orbitals across more than one atom dominated by more than $d$-type functions.
In such cases, it will be difficult to properly couple the unpaired orbitals to a CSF with suitable symmetry in the space of the unoccupied, metal-centered molecular orbitals.

The generation of a spin-coupled CSF requires standard angular momentum coupling as neatly framed in terms of Shavitt graphs of the graphical unitary group approach \cite{Shavitt1981}.
In general, a CSF 
\begin{align}
\label{eq:CSF}
    \ket{\mathbf{t}} = \sum_i d_i \ket{\mathbf{p}}_i
\end{align}
can be represented as a linear combination of Slater determinants $\ket{\mathbf{p}}_i$ with the same orbital configuration and spin-projection quantum number $M$\cite{helgaker_molecular_2000}.
The spin-coupling of a CSF can be represented by a vector $\mathbf{t}$, where each element 
\begin{align}
    t_i = \begin{cases}
     T_i - T_{i-1}, & \text{if $i > 1$},\\
     T_i, & \text{otherwise},
    \end{cases}
\end{align}
where $T_i$ is the total spin obtained from the coupling of the first $i$ electrons. 
Similarly, a determinant can be represented by a vector $\mathbf{p}$, by replacing the total spin with the spin-projection for the elements $T_i$.
Since only open-shell orbitals contribute to the difference of these Slater determinants (a fully closed-shell or high-spin determinant would already represent a CSF),
Slater determinants can be described as a string of creation operators acting on the closed-shell $\ket{cs}$ part of the determinant\cite{helgaker_molecular_2000},
\begin{align}
    \ket{\mathbf{p}} = \prod_i^N \hat{a}_{i, p_i}^\dagger \ket{cs}.
\end{align}
Hence, we can generate a normalized CSF, for an $N$-electron state with a total spin $S$, by the recursive application of an operator $\hat{O}^{S,M}_N$ onto the closed-shell part
\begin{align}
\label{eq:CSF_recursion}
    \ket{\mathbf{t}} & = \hat{O}_N^{S, M}(\mathbf{t}) \ket{cs} \\
    & = \sum_{\sigma = \{-\frac{1}{2}, \frac{1}{2}\}} C_{t_N, \sigma}^{S, M} \hat{O}_{N-1}^{S-t_N, M-\sigma}(\mathbf{t})\hat{a}_{N,\sigma}^\dagger \ket{cs} \\
    & = \sum_{\sigma = \{-\frac{1}{2}, \frac{1}{2}\}} C_{t_N, \sigma}^{S, M} 
        \sum_{\sigma' = \{-\frac{1}{2}, \frac{1}{2}\}} C_{t_{N-1}, \sigma'}^{S-t_N, M-\sigma} 
        \hat{O}_{N-2}^{S-t_N-t_{N-1}, M-\sigma-\sigma'}(\mathbf{t})
        \hat{a}_{N-1,\sigma'}^\dagger \hat{a}_{N,\sigma}^\dagger \ket{cs}
\end{align}
resulting in a linear combination of Slater determinants which are coupled by the Clebsch-Gordon coefficients $C_{t_N, \sigma}^{S, M}$ to normalize the tensor operator.
Both $\alpha$ and $\beta$ orbitals contribute to the coupling so that the recursion on the left hand side of the above equation increases the number of coupled determinants (a high-spin determinant is a special case, as for only $\alpha$ or $\beta$ orbitals contribute).

Let us now consider whether the dominant part of the electron correlation in iron-sulfur cluster originates  from the antiferromagnetic couplings between the iron centers.
One determinant is not sufficient to describe these couplings as it is not an eigenfunction of the $S^2$ operator resulting in unphysical spin-values and CSFs provide a suitable basis as they are eigenfunctions of $S^2$ \cite{Limanni2020,Limanni2021,Dobrautz2021,LiManni_2023, Izsak2023,martidafcik2024spincoupled}.

Application of CSFs in quantum circuits\cite{Sugisaki_2016,Sugisaki2019,Sugisaki_2019a} allows one to extract most information out of a single (broken-symmetry) determinant exploiting spin flips.
It has been shown, that a CSF-based initialization can improve the performance of various quantum algorithms by better capturing strong correlation\cite{martidafcik2024spin}.
Hence, the orbital occupation remains the same, but the spin-orbital occupation changes.
For systems which rely on strong antiferromagnetic couplings such as FeMoco, we artificially divide the static correlation into a part based on these couplings and the remaining correlation based on other effects.

We now analyze whether the application of CSFs can reduce the static correlation on a dinuclear iron-sulfur cluster.
Since this system is less complex than FeMoco, we can efficiently analyze the contributions of the anti-ferromagnetic couplings and other correlation effects.
In [2Fe-2S], both Fe$^{3+}$ centers are in their high-spin configuration with a local spin of $5/2$.
Generated based on the genealogical coupling scheme,
there exist 252 Slater determinants representing this open-shell singlet in the genealogical coupling scheme with only the lowest lying orbitals occupied that then lead to only 42 CSFs.
Obviously, CSFs can efficiently compress the wave function expansion to cover the spin coupling, but the question remains how much electron correlation is then captured given the fact that many Slater determinants are then still neglected from the wave function representation.

\section{Quantum computational preparation of multi- configurational states}
Since the cost $C_{\mathrm{trial}}$ of quantum circuits for preparing multi-configurational states can be significantly more expensive than for the single Hartree--Fock determinant, its feasibility requires some discussion.
Letting $C_{\mathrm{QPE}}$ be the cost of quantum phase estimation, we obtain the ground state with an average of cost $(C_{\mathrm{trial}}+C_{\mathrm{QPE}})/p_{\mathrm{success}}$.
Hence, more expensive multi-configurational states are worth preparing so long as $C_{\mathrm{trial}}\lesssim C_{\mathrm{QPE}}$.
In practically all cases, this inequality is satisfied by a tremendously large margin.
The cost $C_{\mathrm{QPE}}$ scales like $\frac{\mathcal{O}(N^\beta)}{p_\mathrm{success}\epsilon}$ for estimating the ground state energy to a controlled error $\epsilon$.
The exponent $\beta$ for state-of-the-art quantum simulation algorithms~\cite{Low2016HamSim,Low2016Qubitization} is typically between $2$ and $4$, also depending on the details of the molecule. 

In the classically intractable regime around $N=60$ at chemical accuracy $\epsilon=10^{-3}$, previous studies on catalysts for nitrogen and carbon fixation~\cite{vonBurg2020carbon,Lee2020hypercontraction} have narrowed the constant factor in this cost to the order of $C_{\mathrm{QPE}}\sim 10^{10}$ non-Clifford Toffoli gates.
In most cases, the number of non-Clifford gates is a good approximation to the overall spacetime volume cost of executing a quantum circuit on a fault-tolerant quantum computer.
By contrast, Clifford gates are assumed to have negligible cost.
$C_{\mathrm{trial}}$ then refers to the number of primitive non-Clifford gates, such as the Toffoli gate, the $T$ gate, or arbitrary single-qubit rotations. 
We omit logarithmic factor differences in costs (usually at most a factor of 10) between these gates as the class 0, 1, and 2 instances we study will all satisfy $C_{\mathrm{trial}}\ll C_{\mathrm{QPE}}$.
Considering that $C_{\mathrm{trial}}=0$ for single-reference states, we immediately see that preparing sophisticated trial states is an over-looked and under-utilized resource.

We consider three methods of preparing multi-configuration quantum trial states from classically computed trial states. 
First, the method of preparing a superposition of determinants in~\cref{sec:prep_superposition}. Second, the method of preparing a superposition of CSF~\cref{sec:prep_CSF}. Third, the method of preparing an MPS~\cref{sec:prep_MPS}.
We do not consider variational quantum methods~\cite{TILLY20221} as the ansatzes in these methods are typically difficult to initialize and optimize from the classically computed MPS without significant additional quantum cost.

\subsection{Preparing a superposition of determinants}
\label{sec:prep_superposition}
A single determinant wave function in the second-quantized representation is a configuration $\ket{\vec{x}}=\ket{x_0 x_1\cdots x_{2N-1}}$ of $2N$ qubits, where $x_j\in\{0,1\}$ indicates the occupancy of the $j^{\text{th}}$ spin-orbital.
Hence, the number of ones in $\ket{\vec{x}}$ equals the number of electrons.
Starting from the vacuum state $\ket{0}^{\otimes 2N}$, $\ket{\vec{x}}$ is produced using at most $2N$ Pauli $X$ gates

An arbitrary multi-reference configuration of $M$ determinants is then
\begin{align}
\ket{\psi_{\mathrm{M-MC}}}=\sum_{j\in[M]}\alpha_j \ket{\vec{x}_j}.
\end{align} 
This can be prepared in two steps.
First, prepare the superposition state $\ket{\psi}=\sum_{j\in[M]}\alpha_j\ket{j}$ 
where $\ket{j}$ is a binary representation of the integer $j$.
Following Shende et al.~\cite{Shende2006State}, this can be prepared using $O(M)$ non-Clifford gates without any additional qubit beyond that needed to store $\ket{j}$.
This can also be prepared using as few as $\tilde{\mathcal{O}}(\sqrt{M})$ non-Clifford gates but with many ancillary qubit~\cite{Low2018trading}.
Second, let $P$ be a permutation matrix that maps each elements $j\in[M]$ to the binary number represented by $\vec{x}_j$.
Then $\ket{\psi_{\mathrm{M-MC}}}=P\ket{\psi}$.
As $P$ is unitary and a member of the Clifford group, it can be realized as a quantum circuit that uses no non-Clifford gates and $\mathcal{O}(MN)$~\cite{Peng2023Boson} two-qubit Clifford gates.
Hence, $C_{\mathrm{trial}}=\mathcal{O}(M)$.
We easily prepare a state with as many as $M\lesssim C_{\mathrm{QPE}}$ determinants without changing overall runtime by more than a small constant.
As seen in~Table 2 of the main article, 
with $M=N_d^{\mathrm{fin}}$,  $C_{\mathrm{trial}}$ is indeed negligible compared to $C_{\mathrm{QPE}}$ for all cases.

\subsection{Preparing a superposition of CSFs}
\label{sec:prep_CSF}
The wave function can be expressed in a basis of CSFs instead of determinants.
The CSF basis exploits the fact that the total electronic energy in the non-relativistic setting is independent of the angular momentum quantum number $M_S$ (of course, this can be straightfowardly generalized to the relativistic case where spin is no longer a good quantum number owing to spin-orbit coupling and all our conclusions will transfer to relativistic wave functions). 
A CSF initialization has been put forward for 
state preparation in the case of bond breaking processes \cite{martidafcik2024spin}.
For any configuration with spin quantum number $S$, this means it suffices to choose $M_S=S$ for the purpose of maximizing overlap with the (degenerate) space of ground states.
This greatly reduces the number of distinct configuration $\ket{\vec{x}}$ that have to be considered for inclusion into the trial wave function.

We label the basis of CSFs by $\ket{\vec{x}_{\text{p}},\vec{x}_{\text{u}},S}$, where $\vec{x}_{\text{p}}$ indexes fully-filled orbitals, and $\vec{x}_{\text{u}}$ indexes orbitals with unpaired electrons.
Without loss of generality, $\vec{x}_{\text{p}}$ and $\vec{x}_{\text{u}}$ are sorted.
The number of electrons is then $\eta=2|\vec{x}_{\text{p}}|_0+|\vec{x}_{\text{u}}|_0$, of which $|\vec{x}_{\text{u}}|_0$ are unpaired.
Fully-filled orbitals $\vec{x}_{\text{p}}$ correspond a configuration $\ket{\vec{x}}$ for the spin-orbitals where the qubits in positions $2j,2j+1$ are set to one.
In other words, $\ket{x}_{2j}=\ket{x}_{2j+1}=\ket{1}$.
Following the Serber construction~\cite{Sugisaki_2016,Sugisaki2019,Sugisaki_2019a}, let us split $\vec{x}_{\text{u}}=\vec{x}_{\uparrow}\circ \vec{x}_{\text{s}}\circ \vec{x}_{\text{s}'}$ into two parts -- $\vec{x}_{\uparrow}$ is the first $2S$ elements of $\vec{x}_{\text{u}}$ and indexes electrons in the spin-up state and $\vec{x}_{\text{s}}$ is the remainder containing an even number of electrons in the singlet state.
Then given $S$, the spin configuration for the spin state in the basis of $\ket{\vec{x}}$ is chosen to be
\begin{align}
\ket{\vec{x}_{\text{p}},\vec{x}_{\text{u}},S}&=
\left(\bigotimes_{j\notin\vec{x}_{\text{p}},\vec{x}_{\text{u}}}\ket{0}_{2j}\ket{0}_{2j+1}\right)
\otimes
\left(
\bigotimes_{j\in\vec{x}_{\text{paired}}}\ket{1}_{2j}\ket{1}_{2j+1}\right)
\otimes
\left(\bigotimes_{j\in\vec{x}_{\text{t}}}\ket{1}_{2j}\right)
\\\nonumber
&\qquad\otimes\left(\bigotimes^{|\vec{x}_{\text{s}}|_0}_{j=0}\frac{\ket{0}_{2{x}_{\text{s},j}}\ket{1}_{2{x}_{\text{s},j}+1}-\ket{1}_{2{x}_{\text{s},j}}\ket{0}_{2{x}_{\text{s},j}+1}}{\sqrt{2}}\right).
\end{align}
Observe that the singlet state in the last term is a superposition of exponentially many configurations of $\ket{\vec{x}}$.
However, this represents a trivial correlation that can be prepared using $\mathcal{O}(|\vec{x}_{\text{s}}|_0)$ elementary Clifford quantum gates in following sequence
\begin{align}
\ket{0}\ket{0}
\underset{H_0}{\rightarrow} \frac{\ket{0}+\ket{1}}{\sqrt{2}}\ket{0}
\underset{\textsc{CNot}}{\rightarrow} \frac{\ket{0}\ket{0}+\ket{1}\ket{1}}{\sqrt{2}}
\underset{X_0}{\rightarrow} \frac{\ket{0}\ket{1}+\ket{1}\ket{0}}{\sqrt{2}}
\underset{Z_0}{\rightarrow} \frac{\ket{0}\ket{1}-\ket{1}\ket{0}}{\sqrt{2}}.
\end{align}
Hence, there is quantum circuit $U_{\vec{x}_{\text{p}},\vec{x}_{\text{u}},S}\ket{0}^{\otimes{2N}}=\ket{\vec{x}_{\text{p}},\vec{x}_{\text{u}},S}$ that costs $\mathcal{O}(\eta)$ one- and two-qubit Clifford gates.

Of independent interest, it is also possible to prepare a superposition of CSF states, though our main text only studies the overlap of a single CSF state.
We now describe the quantum circuit to prepare such multi-configurational CSF states
\begin{align}
\ket{\psi_{\mathrm{M-CSF}}}=\sum_{j\in[M]}\alpha_j \ket{\vec{x}_{\text{p},j},\vec{x}_{\text{u},j},S_j}.
\end{align}
We define circuits to prepare the superposition state $U\ket{0}_{a}=\sum_{j\in[M]}\sqrt{\frac{\alpha_j}{|\vec{\alpha}|_1}}\ket{j}_{a}$ and $U^{\prime}\ket{0}_{a}=\sum_{j\in[M]}(\sqrt{\frac{\alpha_j}{|\vec{\alpha}|_1}})^{*}\ket{j}_{a}$. 
Following the discussion in~\cref{sec:prep_superposition}, these cost $\mathcal{O}(M)$ non-Clifford gates.
Now, define the multiplexed unitary $U_{\text{CSF}}=\sum_{j\in[M]}\ketbra{j}{j}_a\otimes U_{\vec{x}_{\text{p},j},\vec{x}_{\text{u},j},S_j}$,
which costs $\mathcal{O}(M)$ non-Clifford gates for the controls by $\ket{j}$~\cite{PhysRevX.8.041015} and $\mathcal{O}(M\eta)$ Clifford and non-Clifford gates for all the controlled- $U_{\vec{x}_{\text{p},j},\vec{x}_{\text{u},j},S_j}$.
Observe that the circuit 
\begin{align}
(U^{\prime} \cdot U_{\text{CSF}}\cdot U)\ket{0}_a\ket{0}^{M}=\frac{1}{|\vec{\alpha}|_1}\ket{0}_a\ket{\psi_{\mathrm{M-CSF}}}+\cdots\ket{\perp}_a,
\end{align}
prepares the desired multi-configurational state with probability $\mathcal{O}(1/|\vec{\alpha}|^{2}_1)$.
Hence, using amplitude amplification, we prepare $\ket{\psi_{\mathrm{M-CSF}}}$ with probability one and an overall cost of $\mathcal{O}(M\eta|\vec{\alpha}|_1)$ non-Clifford gates.
In the worst case, we may use a Cauchy--Schwartz inequality to bound $1\le|\vec{\alpha}|_1\le \sqrt{M}|\vec{\alpha}|_2=\sqrt{M}$.
Hence, $C_{\mathrm{trial}}=\mathcal{O}(M^{3/2}\eta)$.
For the case $C_{\mathrm{QPE}}\sim 10^{10}$, we see that the cost of preparing $\ket{\psi_{\mathrm{M-CSF}}}$ is not significant for $M\lesssim 10^{4}$.

\subsection{Preparing a matrix product state}
\label{sec:prep_MPS}
We may also directly prepare the MPS $\ket{\psi_{\mathrm{DMRG}}}$ optimized by DMRG. Recall that we will have an approximate MPS available from the orbital selection procedure of the AutoCAS 
algorithm. Any MPS with bond dimension $m$ may be prepared on a quantum computer with $C_{\mathrm{trial}}=\mathcal{O}(Nm^2)$~\cite{Schon2005MPS} quantum gates in almost constant $\mathcal{O}(m\log\log(N))$ depth~\cite{2023MalzMPS}. A typical converged DMRG solution for class-0 and -1 molecules will have a bond dimension $m\lesssim10^3$, which is much larger than the typical $M< 30$ in $\ket{\psi_{\mathrm{M-MC}}}$. 
However, as the one-time cost of preparing a good multi-configurational state is entirely dominated by the rest of the phase estimation algorithm, it is worthwhile choosing as large a bond dimension as possible if it increases $p_{\mathrm{succees}}$ by any appreciable amount, even if $p_{\mathrm{succees}}$ is already close to one. 

The appropriate MPS to prepare on a quantum computer can be chosen to be the approximate DMRG solutions with modest $m$ (as, for instance obtained by AutoCAS in the active-orbital selection step).
For a reliable orbital selection step, these tend to already have large overlap with the converged DMRG solution with large $m$. 
Note that $C_{\mathrm{trial}}$ for preparing $\ket{\psi_{\mathrm{DMRG}}}$, is sub-dominant to $C_{\mathrm{QPE}}$ even for very large values of $m$. 
For instance, the carbon-fixation catalyst structures in Ref.~\citenum{vonBurg2020carbon} have a phase estimation cost $C_{\mathrm{QPE}}~\sim 10^{10}$ that dominants MPS state preparation for any choice of $m\lesssim 10^4$, or from~\cref{sec:prep_superposition}, preparing a superposition of $M\lesssim 10^{9}$ determinants.

Hence, MPS state preparation is particularly well-suited to class-2 systems.
In these systems, it could be difficult to sample the extremely large number of important determinants, but it could still be easy to compute the approximate DMRG solution.
The MPS overlap of the approximate solution with the converged solution can be close to one, such as in Table 1 of the main article. 
For class-2 systems, we reiterate that while the overlap may still be high in some cases, the QMA-completeness of ground state preparation implies that the hardest cases may have exponentially small overlaps with the true ground state for any constant bond dimension.

\section{Importance of orbital basis selection\label{sec:proof}}

The choice of basis plays a crucial role in obtaining good overlap with the ground state. In particular, the HF overlap can be artificially small with a poor choice of basis. For instance, a random orbital rotation $U(u)$ in the active space with $\eta$ electrons leads to an average overlap of $\int \mathrm{d}u\left|\bra{\psi_{\mathrm{g}}}U(u)\ket{\vec{x}}\right|^2=\binom{2N}{\eta}^{-1}$~\cite{Low2022fermionshadows} for any ground state. This remains exponentially small even with the best multi-configurational state $\ket{\psi_{\mathrm{M-MC}}}$ with fixed $M$. 
In fact, the following \cref{thm:inequality} informs us that the maximum weight of any configuration in a random single-particle basis any virtual space with $N_v$ spin-orbitals and $\eta_v$ electrons is exponentially small with high probability, even if the ground state is single-configuration in some basis. Though a similar fact is easily proven in a random basis on the full space, this is the first time we are aware of a proof in the much more restricted set of a random single-particle basis.
\begin{theorem}[No dominant single-configuration in random single-particle bases] 
\label{thm:inequality}
For any $p^{*}>0$, the probability that any configuration in a random single-particle basis has weight at least $p_{\mathrm{max}}$ is upper bounded by 
\begin{align}
\mathrm{Pr}\left[\max_{\vec{y}}p_{\vec{y}}\ge p_{\mathrm{max}} \ge \sqrt{\frac{\eta_v+1}{p^{*}}/\binom{N_v}{\eta_v}}\right]\le p^{*}
\end{align}
\end{theorem}
This result tells us that with high probability $1-p^{*}\approx 1$ for any $p^{*}$ that is a small constant, say $0.001$, all configurations have weight at most $p_{\mathrm{max}}=\mathcal{O}\left(\sqrt{\eta_v/\binom{N_v}{\eta_v}}\right)$. As $N_v$ and $\eta_v$ increase in tandem, $p_{\mathrm{max}}$ tends to zero exponentially fast. For instance, at half filling, $p_{\mathrm{max}}=\mathcal{O}\left(\frac{N_v^{3/4}}{2^{N_v}}\right)$. 
However, this rapid decay is artificial, and can be avoided by a prudent choice of molecular orbitals. 

In other words, a ground state that is single-configurational in a good orbital basis is in general exponentially multi-configurational in a bad orbital basis. 
We now proceed with the proof of \cref{thm:inequality}.

Consider a single-configurational state $\ket{\vec{x}}$ with $\eta_v$ electrons that occupy $N_v$ spin-orbitals. We assume that core spin-orbitals are fully filled and thus dropped from $\vec{x}$ in the following. Recall that single-particle basis rotations $U(u_v)$ generated by the unitary $u_v\in\mathcal{U}_{N_v}$ drawn from the $N_v\times N_v$ unitary group $\mathcal{U}_{N_v}$ rotate fermion operators like $U(u)a^\dagger_j U^\dagger (u)=\sum_{k}u_{jk} a^\dagger_k$. Now we apply to $\ket{\vec{x}}$ a single-particle basis rotation $U(u_v)$. Here, $u_v$ is uniformly sampled from $\mathcal{U}_{N_v}$ according to the Haar measure. We now prove the following theorem.
\begin{lemma}[Statistics of randomly rotated single-configurations\label{thm:statistics_single_configuration}]
Let $p_{\vec{y}}=\left|\bra{\vec{y}}U(u_v)\ket{\vec{x}}\right|^{2}$ be the probability of observing configuration $\ket{\vec{y}}$ with $\eta_v$ electrons in the rotated single configuration $\ket{\vec{x}}$. Let $u_v$ be sampled uniformly by the Haar measure the unitary group $\mathcal{U}_{N_v}$. Then the mean $\mathbb{E}_{u_v}[p_{\vec{y}}]$ and the covariance $\mathrm{Cov}_{u_v}[p_{\vec{y}},p_{\vec{y}'}]$ and variance $\mathrm{Var}_{u_v}[p_{\vec{y}}]$.are
\begin{align}
\mathbb{E}_{u_v}[p_{\vec{y}}]&=\binom{N_v}{\eta_v}^{-1},
\nonumber \\
\mathrm{Cov}_{u_v}[p_{\vec{y}},p_{\vec{y}'}]&=\frac{\eta_v+1}{\eta_v+1-|\vec{y}\cap\vec{y}'|}\binom{N_v}{\eta_v}^{-2}\left(\frac{|\vec{y}\cap\vec{y}'|}{\eta_v+1}-\frac{\eta_v}{N_v+1}\right),
\nonumber \\
\mathrm{Var}_{u_v}[p_{\vec{y}}]&=(\eta_v+1)\binom{N_v}{\eta_v}^{-2}\left(1-\frac{1}{\eta_v+1}-\frac{\eta_v}{N_v+1}\right)
\le
(\eta_v+1)\binom{N_v}{\eta_v}^{-2}.
\end{align}
\end{lemma}
\begin{proof}
The expectation 
\begin{align}
\mathbb{E}_{u_v}[p_{\vec{y}}]&=\int_{u\sim\mathcal{U}_{N_v}}\mathrm{d}u\mathrm{Tr}\left[\ketbra{\vec{y}}{\vec{y}}U^\dagger(u_v)\ketbra{\vec{x}}{\vec{x}}U(u_v)\right]
\nonumber\\ 
&=\mathrm{Tr}\left[\ketbra{\vec{y}}{\vec{y}}\left\{\int_{u\sim\mathcal{U}_{N_v}}\mathrm{d}u U^\dagger(u_v)\ketbra{\vec{x}}{\vec{x}}U(u_v)\right\}\right]
\nonumber\\  
&=\mathrm{Tr}\left[\ketbra{\vec{y}}{\vec{y}}\left\{I/\binom{N_v}{\eta_v}\right\}\right]
\nonumber\\  
&=\binom{N_v}{\eta_v}^{-1}.
\end{align}
Above, the third line uses the fact that the term in curly brackets has trace $1$ and must be proportional to identity (this may also be proven explicitly following Lemma 11 by Low~\cite{Low2022fermionshadows}).

The covariance is
\begin{align}
\mathrm{Cov}_{u_v}[p_{\vec{y}},p_{\vec{y}'}]
&=\mathbb{E}_{u_v}[(p_{\vec{y}}-\mathbb{E}_{u_v}[p_{\vec{y}}])(p_{\vec{y}'}-\mathbb{E}_{u_v}[p_{\vec{y}'}])]
=\mathbb{E}_{u_v}[p_{\vec{y}}p_{\vec{y}'}]-\binom{N_v}{\eta_v}^{-2}.
\end{align}
The expectation of the cross-term is
\begin{align}
\mathbb{E}_{u_v}[p_{\vec{y}}p_{\vec{y}'}]&=\mathbb{E}_{u_v}[\left|\bra{\vec{y}'}U(u_v)\ket{\vec{x}}\right|^{2}\left|\bra{\vec{y}}U(u_v)\ket{\vec{x}}\right|^{2}]
\nonumber\\ 
&=\int_{u\sim\mathcal{U}_{N_v}}\mathrm{d}u\mathrm{Tr}\left[\left\{\ketbra{\vec{y}}{\vec{y}}\otimes \ketbra{\vec{y}'}{\vec{y}'}\right\} \left\{U^\dagger(u_v)\ketbra{\vec{x}}{\vec{x}}U(u_v)\right\}^{\otimes 2}\right]
\nonumber\\ 
&=\mathrm{Tr}\left[\left\{\ketbra{\vec{y}}{\vec{y}}\otimes \ketbra{\vec{y}'}{\vec{y}'}\right\}\left\{\int_{u\sim\mathcal{U}_{N_v}}\mathrm{d}u \left(U^\dagger(u_v)\ketbra{\vec{x}}{\vec{x}}U(u_v)\right)^{\otimes{2}}\right\}\right]
\nonumber\\  
&=\mathrm{Tr}\left[\left\{\ketbra{\vec{y}}{\vec{y}}\otimes \ketbra{\vec{y}'}{\vec{y}'}\right\}\left\{\sum_{\vec{p},\vec{q}}f(|\vec{p}\cap\vec{q}|)\ketbra{\vec{p}}{\vec{p}}\otimes\ketbra{\vec{q}}{\vec{q}}\right\}\right]
\nonumber\\  
&=\mathrm{Tr}\left[\left\{\ketbra{\vec{y}}{\vec{y}}\otimes \ketbra{\vec{y}'}{\vec{y}'}\right\}f(|\vec{y}\cap\vec{y}'|)\right]
\nonumber\\  
&=f(|\vec{y}\cap\vec{y}'|)
\nonumber\\
&=\frac{\eta_v+1}{\eta_v+1-|\vec{y}\cap\vec{y}'|}\binom{N_v+1}{\eta_v}^{-1}\binom{N_v}{\eta_v}^{-1}.
\end{align}
Above, the fourth line follows from Theorem.~8 by Low~\cite{Low2022fermionshadows}, where $|\vec{p}\cap\vec{q}|$ counts the number of common `one' elements in the same position in both $\vec{p}$ and $\vec{q}$, and the function $f(x)=\frac{\eta_v+1}{\eta_v+1-x}\binom{N_v+1}{\eta_v}^{-1}\binom{N_v}{\eta_v}^{-1}$.
Hence, the covariance 
\begin{align}
\label{eq:covariance}
\mathrm{Cov}_{u_v}[p_{\vec{y}},p_{\vec{y}'}]&=\frac{\eta_v+1}{\eta_v+1-|\vec{y}\cap\vec{y}'|}\binom{N_v+1}{\eta_v}^{-1}\binom{N_v}{\eta_v}^{-1}-\binom{N_v}{\eta_v}^{-2}
\nonumber \\
&=\frac{\eta_v+1}{\eta_v+1-|\vec{y}\cap\vec{y}'|}\binom{N_v}{\eta_v}^{-2}\left(1-\frac{\eta_v}{N_v+1}-\frac{\eta_v+1-|\vec{y}\cap\vec{y}'|}{\eta_v+1}\right)
\nonumber \\
&=\frac{\eta_v+1}{\eta_v+1-|\vec{y}\cap\vec{y}'|}\binom{N_v}{\eta_v}^{-2}\left(\frac{|\vec{y}\cap\vec{y}'|}{\eta_v+1}-\frac{\eta_v}{N_v+1}\right).
\end{align}
The variance follows from setting $\vec{y}=\vec{y}'$ and observing that $|\vec{y}\cap\vec{y}'|=\eta_v$.
\end{proof}

We note in passing that a similar result for averaging over a Haar random unitary $U\sim\mathcal{U}_{\binom{N_v}{\eta_v}}$ on the full space is well-known. 
Compared to averaging over the more restricted single-particle rotations $U(u),\;u\sim\mathcal{U}_{N_v}$, that result has the same expectation, but zero covariance and a variance $\frac{2}{\binom{N_v}{\eta_v}(\binom{N_v}{\eta_v}+1)}$, which is $\mathcal{O}(\eta)$ smaller.
We now bound the probability that the maximum of $p_{\vec{y}}$ over $\vec{y}$ is larger than some value, say $p_{\mathrm{max}}$. 
We do so with the help of the following lemma.

\begin{lemma}[A one-sided inequality of the Chebyshev from Theorem. 4 by Marshall and Olkin~\cite{1960MarshallOneSidedChebyshev}\label{thm:one_sided}]
If $x_1,\cdots,x_k$ are random variables with $\mathbb{E}[x_i]=0$, $\mathbb{E}[x_i^2]=\sigma_i^2$, $i=1,\cdots,k$ then
\begin{align}
\mathrm{Pr}\left[\max_{i}x_i\ge 1\right]\le \sum_{i=1}^k\frac{\sigma_i^2}{1+\sigma_i^2}=\sum_{i=1}^k\frac{1}{1+\sigma_i^{-2}}.
\end{align}
\end{lemma}
We substitute the means and variances of \cref{thm:statistics_single_configuration} into \cref{thm:one_sided} together with some simple substitutions to obtain the proof of our main result.
\begin{proof}[Proof of \cref{thm:inequality}]
\begin{align}
\mathrm{Pr}\left[\max_{\vec{y}}p_{\vec{y}}\ge p_{\mathrm{max}}\right]
&\le \sum_{\vec{y}}\frac{1}{1+(p_{\mathrm{max}}-\mathbb{E}[p_{\vec{y}}])^2/\mathrm{Var}[p_{\vec{y}}]}
\nonumber \\
&=\binom{N_v}{\eta_v}\frac{1}{1+(p_{\mathrm{max}}-\mathbb{E}[p_{\vec{y}}])^2/\mathrm{Var}[p_{\vec{y}}]}
\nonumber \\
&=\binom{N_v}{\eta_v}\frac{1}{1+\binom{N_v}{\eta_v}^{2}(p_{\mathrm{max}}-\mathbb{E}[p_{\vec{y}}])^2/(\eta_v+1)}
\nonumber \\
&=\frac{1}{\binom{N_v}{\eta_v}^{-1}+\binom{N_v}{\eta_v}(p_{\mathrm{max}}-\mathbb{E}[p_{\vec{y}}])^2/(\eta_v+1)}.
\end{align}
Let us choose $p_{\mathrm{max}}=\mathbb{E}[p_{\vec{y}}]+\sqrt{\frac{\eta_v+1}{p^{*}}/\binom{N_v}{\eta_v}} \ge \sqrt{\frac{\eta_v+1}{p^{*}}/\binom{N_v}{\eta_v}}$, where $p^{*}>0$. Then
\begin{align}
\mathrm{Pr}\left[\max_{\vec{y}}p_{\vec{y}}\ge \sqrt{\frac{\eta_v+1}{p^{*}}/\binom{N_v}{\eta_v}}\right]
\le
\frac{1}{\binom{N_v}{\eta_v}^{-1}+1/p^{*}}\le p^{*}.
\end{align}
\end{proof}
\clearpage

\end{appendix}

\newpage
\providecommand{\latin}[1]{#1}
\makeatletter
\providecommand{\doi}
  {\begingroup\let\do\@makeother\dospecials
  \catcode`\{=1 \catcode`\}=2 \doi@aux}
\providecommand{\doi@aux}[1]{\endgroup\texttt{#1}}
\makeatother
\providecommand*\mcitethebibliography{\thebibliography}
\csname @ifundefined\endcsname{endmcitethebibliography}
  {\let\endmcitethebibliography\endthebibliography}{}

\end{document}